\documentclass[10pt]{sigplanconf}

\usepackage{amsmath}

\usepackage[usenames,dvipsnames]{color}

\usepackage{mathpartir}

\usepackage{amsthm}
\usepackage{stmaryrd}

\usepackage{listings}

\usepackage[colorinlistoftodos]{todonotes}

\usepackage{booktabs}

\newcommand{\webpageTAJS}{\texttt{http://www.brics.dk/TAJS/}}

\newcommand{\webpageGOOGLE}{\texttt{http://v8.googlecode.com/svn/data/benchmarks/v7/run.html}}
\newcommand{\webpageCOOKIE}{\texttt{http://code.google.com/p/cookies/}}
\newcommand{\webpageRYE}{\texttt{http://ryejs.com/}}

\newtheorem{definition}{Definition}
\newtheorem{lemma}{Lemma}
\newtheorem{theorem}{Theorem}

\newcommand{\Label}[1]{\textit{#1}}

\newcommand{\eval}{\Downarrow}

\newcommand{\entails}{\vdash}
\newcommand{\entailsP}{\entails_{\textbf{P}}^{\aALatticeR,\aALatticeQ}}

\newcommand{\entailsFAIteration}{\entails_{\textsc{App}}^{\aStorables}}
\newcommand{\entailsFAApplication}{\entails_{\textsc{App}}^{\aFunc}}

\newcommand{\entailsPRIteration}{\entails_{\textsc{Get}}^{\aStorables}}
\newcommand{\entailsPRIntersection}{\entails_{\textsc{Get}}^{\aObj}}

\newcommand{\entailsPAIteration}{\entails_{\textsc{Put}}^{\aObjlabel}}
\newcommand{\entailsPAAssignment}{\entails_{\textsc{Put}}^{\ljSLocation}}

\newcommand{\join}{\sqcup}

\newcommand{\equivType}{\equiv_{\bijection,\ljType}}
\newcommand{\equivTypeSub}{\equiv_{\bijection,\{\ljSource|\ljSource\notin\bar{\ljType}\}}}
\newcommand{\equivTypeSubPrime}{\equiv_{\bijection',\{\ljSource|\ljSource\notin\bar{\ljType}\}}}

\newcommand{\equivLattice}{\prec_{\aAnalysisLattice}}

\newcommand{\bijection}{\flat}

\newcommand{\powerset}{\wp}

\newcommand{\dom}{dom}

\newcommand{\ljOperation}{\Downarrow_{\ljOp}^{\ljVal}}
\newcommand{\abstractOperation}{\Downarrow_{\ljOp}^{\aVal}}

\newcommand{\ljNew}{\textbf{new}^{\ljSLocation}}
\newcommand{\ljIf}{\textbf{if}}

\newcommand{\ljOp}{\textbf{op}}

\newcommand{\ljStr}{\textit{str}}

\newcommand{\ljTrue}{true}
\newcommand{\ljFalse}{false}
\newcommand{\ljUndefined}{\textbf{undefined}}
\newcommand{\ljNull}{\textbf{null}}

\newcommand{\ljTrace}{\textbf{trace}^{\ljSource}}

\newcommand{\ljTraceClass}{\textbf{trace}^{\ljSource,\ljClass}}
\newcommand{\ljUntraceClass}{\textbf{untrace}^{(\ljClass\hookrightarrow\ljClass')}}

\newcommand{\LjValue}{\Label{Value}}
\newcommand{\LjTaintedValue}{\Label{Tainted Value}}

\newcommand{\LjPrototype}{\Label{Prototype}}
\newcommand{\LjObject}{\Label{Object}}

\newcommand{\LjClosure}{\Label{Closure}}
\newcommand{\LjLocation}{\Label{Location}}
\newcommand{\LjStorable}{\Label{Storable}}
\newcommand{\LjEnvironment}{\Label{Environment}}
\newcommand{\LjHeap}{\Label{Heap}}

\newcommand{\lj}{\lambda_{J}}
\newcommand{\ldj}{\lambda_{J}^{\D}}
\newcommand{\ldcj}{\lambda_{J}^{\ljClass}}

\newcommand{\ljHeap}{\mathcal{H}}
\newcommand{\ljEnv}{\rho}

\newcommand{\ljExp}{e}
\newcommand{\ljLocation}{\xi^{\ljSLocation}}
\newcommand{\ljLocationPrime}{\xi^{\ljSLocation\prime}}
\newcommand{\ljSLocation}{\ell}
\newcommand{\ljSource}{\ljSLocation}

\newcommand{\ljFunc}{\lambda^{\ljSLocation}}

\newcommand{\ljProto}{p}
\newcommand{\ljObj}{o}
\newcommand{\ljClosure}{f}
\newcommand{\ljStorable}{s}
\newcommand{\ljVal}{v}
\newcommand{\ljConst}{c}
\newcommand{\ljVar}{x}
\newcommand{\ljTypedVal}{\omega}
\newcommand{\ljType}{\kappa}

\newcommand{\joinType}{\bullet}

\newcommand{\ljClass}{\mathcal{A}}

\newcommand{\D}{\mathcal{D}}

\newcommand{\dTrace}{\tau}

\newcommand{\aVal}{\vartheta}
\newcommand{\aValBottom}{\aVal_{\perp}}

\newcommand{\aState}{\Gamma}
\newcommand{\aStateBottom}{\aState_{\perp}}

\newcommand{\aScope}{\sigma}
\newcommand{\aScopeBottom}{\aScope_{\perp}}

\newcommand{\aStorable}{\theta}
\newcommand{\aStorables}{\Theta}

\newcommand{\aObjlabel}{\Xi}
\newcommand{\aObjStore}{\Sigma}

\newcommand{\aObj}{\Delta}

\newcommand{\aFunc}{\Lambda^{\ljSLocation}}
\newcommand{\aFuncBottom}{\aFunc_{\perp}}

\newcommand{\aFuncInput}{In}
\newcommand{\aFuncOutput}{Out}

\newcommand{\aStore}{\mathcal{F}}
\newcommand{\aStoreBottom}{\aStore_{\perp}}

\newcommand{\AAnalysisLattice}{\Label{Analysis Lattice}}
\newcommand{\aAnalysisLattice}{\mathcal{C}}

\newcommand{\aALatticeR}{\mathcal{R}}
\newcommand{\aALatticeQ}{\mathcal{Q}}

\newcommand{\aLattice}{\mathcal{L}}
\newcommand{\aLatticeBottom}{\aLattice_{\perp}}

\newcommand{\LvUndefined}{\textit{Undefined}}
\newcommand{\LvNull}{\textit{Null}}
\newcommand{\LvBool}{\textit{Bool}}

\newcommand{\LvNum}{\textit{Num}}

\newcommand{\LvString}{\textit{String}}

\newcommand{\lvUndefined}{\mathtt{undefined}}
\newcommand{\lvTrue}{\mathtt{true}}
\newcommand{\lvFalse}{\mathtt{false}}
\newcommand{\lvBool}{\mathtt{bool}}
\newcommand{\lvNull}{\mathtt{null}}

\newcommand{\lvString}{\mathtt{str}}

\newcommand{\ALatticeValue}{\Label{Lattice Value}}

\newcommand{\ALabel}{\Label{Label}}
\newcommand{\AValue}{\Label{Abstract Value}}
\newcommand{\AClosure}{\Label{Abstract Closure}}
\newcommand{\AObject}{\Label{Abstract Object}}
\newcommand{\AStorable}{\Label{Abstract Storable}}
\newcommand{\AFunctionStore}{\Label{FunctionStore}}
\newcommand{\AScope}{\Label{Scope}}
\newcommand{\AStorableStore}{\Label{ObjectStore}}
\newcommand{\AState}{\Label{State}}

\newcommand{\DDependency}{\Label{Dependency}}

\newcommand{\Rule}[1]{\textsc{#1}}

\newcommand{\LDJConstant}{(DT-Const)}
\newcommand{\LDJVariable}{(DT-Var)}
\newcommand{\LDJOperation}{(DT-Op)}
\newcommand{\LDJObjectCreation}{(DT-New)}
\newcommand{\LDJFunctionCreation}{(DT-Abs)}
\newcommand{\LDJFunctionApplication}{(DT-App)}
\newcommand{\LDJPropertyReference}{(DT-Get)}
\newcommand{\LDJPropertyAssignment}{(DT-Put)}
\newcommand{\LDJConditionTrue}{(DT-IfTrue)}
\newcommand{\LDJConditionFalse}{(DT-IfFalse)}
\newcommand{\LDJTrace}{(DT-Trace)}

\newcommand{\LDJTraceExt}{(DT-Trace-Classified)}
\newcommand{\LDJUntrace}{(DT-Untrace)}

\newcommand{\AProgram}{(A-Program)}
\newcommand{\AProgramIterationNotEquals}{(P-Iteration-NotEquals)}
\newcommand{\AProgramIterationEquals}{(P-Iteration-Equals)}

\newcommand{\AConstant}{(A-Const)}
\newcommand{\AVariable}{(A-Var)}
\newcommand{\AOperation}{(A-Op)}
\newcommand{\AObjectCreation}{(A-New-NonExisting)}
\newcommand{\AObjectCreationExisting}{(A-New-Existing)}

\newcommand{\AFunctionCreation}{(A-Abs-NonExisting)}
\newcommand{\AFunctionCreationExisting}{(A-Abs-Existing)}
\newcommand{\AFunctionApplication}{(A-App)}
\newcommand{\AFunctionIteration}{(App-Iteration)}
\newcommand{\AFunctionIterationEmpty}{(App-Iteration-Empty)}
\newcommand{\AFunctionStoreSubset}{(App-Store-Subset)}
\newcommand{\AFunctionStoreNonSubset}{(App-Store-NonSubset)}

\newcommand{\APropertyReference}{(A-Get)}
\newcommand{\APropertyReferenceIteration}{(Get-Iteration)}
\newcommand{\APropertyReferenceIterationEmpty}{(Get-Iteration-Empty)}
\newcommand{\APropertyReferenceIntersection}{(Get-Intersection)}
\newcommand{\APropertyReferenceNonIntersection}{(Get-NonIntersection)}

\newcommand{\APropertyReferenceEmpty}{(Get-Empty)}

\newcommand{\APropertyAssignment}{(A-Put)}
\newcommand{\APropertyAssignmentIteration}{(Put-Iteration)}
\newcommand{\APropertyAssignmentIterationEmpty}{(Put-Iteration-Empty)}
\newcommand{\APropertyAssignmentInDom}{(Put-Assignment-InDom)}
\newcommand{\APropertyAssignmentNotInDom}{(Put-Assignment-NotInDom)}

\newcommand{\AConditionTrue}{(A-IfTrue)}
\newcommand{\AConditionFalse}{(A-IfFalse)}

\newcommand{\ACondition}{(A-If)}
\newcommand{\ATrace}{(A-Trace)}

\newcommand{\ATraceExt}{(A-Trace-Classified)}
\newcommand{\AUntrace}{(A-Untrace)}

\lstdefinelanguage{JavaScript}{
                keywords={      attributes, class, classend, do, empty, endif, endwhile, fail,
                                function, functionend, if, implements, in, inherit, inout, not, of,
                                operations, out, return, set, then, types, while, use, else, switch, case,
                break, default, for, var},
                keywordstyle=\color{blue}\bfseries,
                ndkeywords={trace, untrace},
                ndkeywordstyle=\color{red}\bfseries,
                identifierstyle=\color{black},
                sensitive=false,
                comment=[l]{//},
                commentstyle=\color{Gray}\ttfamily,
                stringstyle=\color{red}\ttfamily
}
\newcommand{\lvNUM}{\textit{NUM}}
\newcommand{\lvCHAR}{\textit{STRING}}

\begin{document}

\conferenceinfo{PLAS'13,} {June 20, 2013, Seattle, WA, USA.}
\CopyrightYear{2013}
\copyrightdata{978-1-4503-2144-0/13/06}

\preprintfooter{Type-based Dependency Analysis for JavaScript}

\title{Type-based Dependency Analysis for JavaScript}
\subtitle{Technical Report}

\authorinfo{Matthias Keil \and Peter Thiemann}
{Institute for Computer Science\\ University of Freiburg\\ Freiburg, Germany}
{\{keilr,thiemann\}@informatik.uni-freiburg.de}

\maketitle

\begin{abstract}
Dependency analysis is a program analysis that determines potential data flow between program points. While it is not a security analysis per se, it is a viable basis for investigating data integrity, for ensuring confidentiality, and for guaranteeing sanitization. A noninterference property can be stated and proved for the dependency analysis. 

We have designed and implemented a dependency analysis for JavaScript. We formalize this analysis as an abstraction of a tainting semantics. We prove the correctness of the tainting semantics, the soundness of the abstraction, a noninterference property, and the termination of the analysis. 
\end{abstract}


\category{F.3.2}{LOGICS AND MEANINGS OF PROGRAMS}{Semantics of Programming Languages}[Program analysis]
\category{D.3.1}{PROGRAMMING LANGUAGES }{Formal Definitions and Theory}[Semantics ]
\category{D.4.6 }{OPERATING SYSTEMS}{Security and Protection}[Information flow controls ]

\terms
Security

\keywords
Type-based Analysis, Dependency, JavaScript




\section{Introduction} \label{sec:introduction}

Security Engineering is one of the challenges of modern software development. The connected world we live
in consists of interacting entities that process distributed private data. This data has to be protected against illegal
usage, tampering, and theft.


Web applications are one popular example of such interacting entities. They run in the web browser and consist of
program fragments from different sources (e.g., mashups). Such fragments should not be entrusted with sensitive
information. However, if a fragment's input data can be shown not to depend on confidential data, then it cannot divulge
this data or tamper with it.

A Web application may also be vulnerable to an injection attack. Such
an attack arises when data is stored in a database or in the DOM without
proper escaping. If an analysis can determine that the input to the database never
depends directly on a data source (like an HTML input field), but
rather is always filtered by a suitable sanitizer, then many kinds of
injection attacks can be avoided.

Dependency analysis is a program analysis that can help in both
situations, because it determines potential data flow between program
points. Intuitively, there is a dependency between the value in
variable $\ljVar$ and the value of an expression $\ljExp[\ljVar]$
containing the variable if substituting different expressions
$\ljExp'$ for $\ljVar$ may change the value of $\ljExp[\ljExp']$. In
our application we label data sources (e.g., as confidential) in a
JavaScript program and are interested in identifying the potential
sinks reachable from these data sources. For sanitization, we
instrument sanitizers with a relabeling operation that modifies the
dependency on the original data source to a sanitized dependency. We
consider a data sink safe, if it only depends on data that passed
through a sanitizer. Other uses of dependency
information for optimization or parallelization are possible, but not
considered in this work.

We designed and implemented our dependency analysis as an
extension of TAJS \cite{tajs2009}, a type analyzer for JavaScript. The
implementation allows us to label data sources with a $\ljTrace$ marker and to
indicate relabelings with an $\ljUntraceClass$ marker. The analysis
performs an abstract interpretation to approximate the flow of the
markers throughout the program. The marker is part of the analyzed
type and propagated to all program points that depend on a marked
value.

Part of our work consists in establishing the formal
underpinnings of the implemented analysis. 
Thus, we outline a correctness proof for the dependency part of the
analysis. For conducting the proofs, we have simplified the domains
with respect to the implementation to avoid an overly complex formal
system.  To this end, we formalize the dynamic semantics of a
JavaScript core language, extend that with marker propagation, and
then formalize the abstract interpretation of this extended semantics.
Both, concrete and abstract semantics are given as big-step semantics.
We prove sound marker propagation, sound approximation of the dynamic semantics by its abstract counterpart,
noninterference, and the termination of the analysis. 

\paragraph{Contributions}
\begin{itemize}
\item Design and implementation of a type-based dependency analysis based on TAJS \cite{tajs2009}.
\item Formalization of the analysis.
\item Proofs of correctness and termination.
\item Extension of the analysis for sanitization (Section~\ref{sec:untrace}).
\item A noninterference theorem (Section~\ref{sec:technical_results}).
\end{itemize}
\paragraph{Overview}
Section~\ref{sec:application_scenario} considers some example scenarios of our implemented system.
Section~\ref{sec:formalization} formalizes a core language, its
dynamic semantics, and defines noninterference
semantically. Section~\ref{sec:dependency_type} extends this semantics with tainting. 
Section~\ref{sec:abstract_analysis} defines the corresponding
abstraction, Section~\ref{sec:applying-analysis} gives some example applications, and Section \ref{sec:untrace}
defines the extension for sanitization. Section~\ref{sec:technical_results} contains our
theorems of soundness, noninterference, and termination. Section~\ref{sec:implementation} briefly describes the
implementation. Related work is discussed in
Section~\ref{sec:related_work} followed by a conclusion.


The proofs of soundness, noninterference, and termination are shown in appendix \ref{sec:proof_context},
\ref{sec:proof_noninterference}, \ref{sec:proof_correctness}, and \ref{sec:proof_termination}.


\section{Application Scenario} \label{sec:application_scenario}

Web developers rely on third-party libraries for calendars, maps,
social networks, and so on. To create a
trustworthy application, they should ensure that these libraries do
not leak sensitive information of their users.

One way to avoid such leaks is to detect information flow from
confidential data sources to untrusted sinks by program
analysis and take
measures to avoid this flow. Sometimes, this approach is too
restrictive, because the data arriving at the sink has been sanitized
on the way from the source. Sanitization can take many forms: data may
have been encrypted or a username/password combination may have been
reduced to a boolean. In such cases, the resulting data still depends
on the confidential source, but it can be safely declassified and passed
on to an untrusted sink.

An analogous scenario is the avoidance of injection attacks
where direct dependencies of database queries or DOM contents from
input fields in a Web form should be avoided. However, an indirect
dependency via a sanitizer that, in this case, escapes the values
suitably is acceptable.

Our dependency analysis addresses both scenarios as illustrated with
the following examples.

\subsection{Cookies}
\label{sec:cookies}

A web developer might want to ensure that the code does not read sensitive
data from cookies and sends it to the net. Technically, it means
that data that is passed to network send operations must not depend on
\lstinline{document.cookie}. This dependency can be checked by our analysis.  

To label data sources, our implementation reads a configuration file
with  a list of JavaScript objects that are labeled with a
dependency mark before starting the analysis.
Any predefined value or function can be marked in this way.
For this example, the analyzer is to label
\lstinline{document.cookie} with \lstinline{t0}. 

Values that are written
to a cookie are labeled by wrapping them in a \lstinline{trace}
expression.
The analysis determines that values returned
from cookies are influenced by
\lstinline{document.cookie}. Furthermore, after writing a marked value 
to a cookie, 
each subsequent read operation returns a value that depends on it. 

The following code snippet uses
a standard library for reading and writing cookies. The comments
show the analyzed dependencies of the respective values. 

\begin{lstlisting}[language=JavaScript]
var val1 = readCookie('test'); // d(val1)={t0}
var val2 = trace(4711);        // d(val2)={t1}
writeCookie('test', val2);
var val3 = readCookie('test'); // d(val3)={t0,t1}
\end{lstlisting}




Thus, the read value in \lstinline{val1} is influenced by \lstinline{document.cookie}. The value
in \lstinline{val2} is labeled by a fresh mark \lstinline{t1}. Later, this value is written to the cookie. Hence, the
result \lstinline{val3} of the last read operation is influenced by \lstinline{document.cookie} and \lstinline{val2}.


\subsection{Application: Sensitive Data}
\label{sec:application_sensitive_data}

The next example is to illustrate the underpinnings of our analysis and to
point out differences to other techniques. 
Figure~\ref{lst:example1} shows a code fragment to request
the user name corresponding to a user
id. This data is either read from a cookie or obtained by an Ajax
request.


\begin{figure}
\begin{lstlisting}[language=JavaScript]
var userHandler = function(uid) {
   var userData = {name:''};
   var onSuccess = function(response) {
      userData = response;
   };

   if (Cookie.isset(uid)) {
      Cookie.request(uid, onSuccess);
   } else {
      Ajax.request('http:\\example.org', {
         content : uid
      }, onSuccess);
   }

   return {
      getName : function() {
         return userData.name;
      }
   }
};
var name1 = userHandler(trace("uid1")).getName();
var name2 = userHandler(trace("uid2")).getName();
\end{lstlisting}
                \caption{Loading sensitive data.}
                \label{lst:example1}
\end{figure}

The function \lstinline{userHandler} returns an interface to a
user's personal data. The implementation abstracts from the data
source by using a callback  function \lstinline{onSuccess} to handle
the results. The code ignores the problem that \lstinline{userData}
may not be valid before completion of the Ajax request. 

To detect all values depending on user information, a developer would
mark the id. This mark should propagate to values returned from
\lstinline{Cookie.request()} and \lstinline{Ajax.request()}. Because
we are interested in values depending on \lstinline{Cookie.request()}
and \lstinline{Ajax.request()} the interfaces also get marked.

The conditional in line~7 depends on \lstinline{Cookie.isset(uid)} and
thus on \lstinline{uid} and on the cookie interface. The 
value in \lstinline{userData} (line~4) depends on \lstinline{uid}, on
the cookie interface, and on
\lstinline{Cookie.isset(uid)} from the Ajax interface. The result \lstinline{name1} depends on 
\lstinline{userData.name} and therefore on the user id, the cookie
interface, and potentially on the Ajax interface.

Standard security analyses label values with marks
drawn from a security lattice, often just \emph{Low} and \emph{High}. If
both sources, the cookie interface 
and the Ajax interface, are labeled with the same mark, there is no way to
distinguish these sources. Dependencies allow us a to
formulate security properties on a fine level of granularity that
distinguishes different sources without changing the underlying lattice.

Second, our analysis is flow-sensitive. Dependencies are bound to
values instead of variable names or parameters. A
variable may containt different values depending on different sources
during evaluation. In addition, the underlying TAJS implementation
already handles aliasing and polyvariant analysis in a
satisfactory way.

In the example, the values in \lstinline{name1} and \lstinline{name2}
result from the same function but may depend on different sources. 
The flow-sensitive model retains the independence of  the value in \lstinline{name1} and
\lstinline{trace("uid2")}. 
Section~\ref{sec:application_sensitive_data_cont} discusses the actual
outcome of the analysis.


\subsection{Application: Foreign Code}
\label{sec:application_foreign_code}


\begin{figure}
\begin{lstlisting}[language=JavaScript]
loadForeignCode = trace(function() {
   Array.prototype.foreach = function(callback) {
      for ( var k = 0; k < this.length; k++) {
         callback(k, this[k]);
      }
   };
});
loadForeignCode();
// [..]
var array = new Array(4711, 4712);
array.foreach(function(k, v) {
   result = k + v;
});
\end{lstlisting}
                \caption{Using foreign Code.}
                \label{lst:example2}
\end{figure}

The second scenario (Figure~ \ref{lst:example2}) illustrates one way a library can extend existing functionality.
This example extends the prototype of \lstinline{Array} by a
\lstinline{foreach} function. Later on, this function is used to
iterate over elements.  

The goal here is to
protect code from being compromised by the libraries used. The
function \lstinline{loadForeignCode} encapsulates foreign code and is
labeled as a source. In consequence, all values created or modified
by calling \lstinline{loadForeignCode} depend on this function and
contain its mark. Because the function in the \lstinline{foreach}
property gets marked, the values in \lstinline{result} also get
marked. Therefore, \lstinline{result} may be influenced by loading
foreign code. See Section~\ref{sec:application_foreign_code_cont} for the
results of the analysis.

\subsection{Application: Sanitization}
\label{sec:application:-sanitization}

Noninterference is not the only interesting property that can be
investigated with the dependency analyzer. To avoid injection attacks,
programmers should ensure that only escaped values occur in a database
query or become part of an HTML page. Also, a dependency on a secret data source may be acceptable
if the data is encrypted before being published.  These examples
illustrate the general idea of sanitization where a suitable function
needs to be interposed in the dataflow between certain sources and
sinks.


\begin{figure}
\begin{lstlisting}[language=JavaScript]
$ = function(id) {
  return trace(document.getElementById(id).value, "#DOM");
}
function sanitizer(value) {
  /* clean up value ... */
  return untrace(value, "#DOM");
}
// [...]
var input = $("text");
var secureInput = sanitizer(input);
consumer(secureInput);
\end{lstlisting}
                \caption{Analyzing sanitization.}
                \label{listing:sanitization}
\end{figure}

The concrete example in Figure~\ref{listing:sanitization} applies our analysis to the
problem. The input is
labeled with mark  \lstinline{#DOM} (line~2). 
The function in line 4 performs some (unspecified) sanitization and
finally applies the \lstinline'untrace' function to mark the
dependency on the marks identified with \lstinline{#DOM} as a
sanitized, safe dependency. The argument of the consumer can now be
checked for dependencies on unsanitized values. In the example code,
the analysis determines that the argument depends on the DOM, but that
the dependency is sanitized. 

Changing line~10 as indicated below
leaves the argument of the consumer with a mixture of sanitized
and unsanitized dependencies. This mixture could be flagged as an error.
\begin{lstlisting}[language=JavaScript,numbers=none]
var secureInput =
  i_know_what_i_do ? sanitizer(input) : input;
\end{lstlisting}


\section{Formalization} \label{sec:formalization}

This section presents the JavaScript core calculus $\lj$ along with a
semantic definition of independence.


\subsection{Syntax of $\lj$} \label{sec:syntax_ldj}

\begin{figure}[t]
  \begin{displaymath}
    \begin{array}{lrl}
      \ljExp &::=& \ljConst \mid \ljVar \mid \ljFunc \ljVar.\ljExp \mid \ljExp(\ljExp) \mid \ljOp(\ljExp, \ljExp) \\
      &\mid& \ljIf~ (\ljExp)~ \ljExp,~ \ljExp \mid \ljNew~\ljExp \mid \ljExp[\ljExp] \mid \ljExp[\ljExp]=\ljExp \mid \ljTrace(\ljExp)
    \end{array}
  \end{displaymath}
  \begin{displaymath}
    \begin{array}{llrl}
      \LjLocation &\ni~\ljLocation &&\\       
      \LjValue &\ni~ \ljVal &::=& \ljConst ~|~ \ljLocation\\
      \LjPrototype &\ni~\ljProto &::=& \ljVal\\
      \LjClosure &\ni~\ljClosure &::=& \emptyset ~|~ \langle \ljEnv, \ljFunc \ljVar.\ljExp \rangle\\
      \LjObject &\ni~ \ljObj &::=& \emptyset ~|~ \ljObj[\ljStr\mapsto\ljVal]\\
      \LjStorable &\ni~ \ljStorable &::=& \langle \ljObj, \ljClosure, \ljProto \rangle\\
      \LjEnvironment &\ni~ \ljEnv &::=& \emptyset ~|~ \ljEnv[\ljVar\mapsto\ljVal]\\
      \LjHeap &\ni~ \ljHeap &::=& \emptyset ~|~ \ljHeap[\ljLocation\mapsto\ljStorable]
    \end{array}
  \end{displaymath}
  \caption{Syntax and semantic domains of $\lj$.}
  \label{fig:semantic-domains_lj}  \label{fig:syntax_lj}
\end{figure}

$\lj$ is inspired by JavaScript core calculi from the literature
\cite{Guha:2010:EJ:1883978.1883988,ecma1999:262}.
A $\lj$ expression (Figure \ref{fig:syntax_lj}) is either a constant $\ljConst$ (a boolean, a
number, a string, $\ljUndefined$, or $\ljNull$), a variable $\ljVar$,
a lambda expression, an application, a primitive operation, 
a conditional, an object creation, a property reference, a property assignment, or a trace
expression.

The trace expression is novel to our calculus. It creates marked
values that can be tracked by our dependency analysis.
The expression $\ljNew~\ljExp$ creates an object whose prototype is the result of $\ljExp$.
The lambda expression, the new expression, and the trace expression
carry a unique mark $\ljSLocation$.


\subsection{Semantic domains} \label{sec:semantic_domains}

Figure \ref{fig:semantic-domains_lj} also defines the semantic domains of $\lj$.  A heap maps a location
$\ljLocation$ to a storable $\ljStorable$, which is a triple consisting of an object $\ljObj$,
potentially a function closure $\ljClosure$ (only for function objects), and a value $\ljProto$,
which serves as the prototype. The superscript $\ljSLocation$ refers the expression causing the 
allocation. An object $\ljObj$ maps a string to a value. A closure
consists of an environment $\ljEnv$ and an expression $\ljExp$.  The environment $\ljEnv$ maps a
variable to a value $\ljVal$, which may be a base type constant or a location.

Program execution is modeled by a big-step evaluation judgment of the form $\ljHeap,\ljEnv
~\entails~\ljExp ~\eval~ \ljHeap' ~|~ \ljVal$. The evaluation of expression $\ljExp$ in an initial heap
$\ljHeap$ and environment $\ljEnv$ results in the final heap $\ljHeap'$ and the value $\ljVal$. We omit
its standard definition for space reasons, but show excerpts of an augmented semantics in 
Section~\ref{sec:dependency_type}.

\begin{figure}[t]
                \centering
                \begin{displaymath}
                                \begin{array}{l@{~}l@{~}l}

                                                \langle \ljObj, \ljClosure, \ljProto \rangle(\ljStr) ~&::=&~ \begin{cases}
                                                                \ljVal, & \ljObj=\ljObj'[\ljStr\rightarrow\ljVal]\\
                                                                \ljObj'(\ljStr), & \ljObj=\ljObj'[\ljStr'\rightarrow\ljVal] \\ 
                                                                \ljHeap(\ljLocation)(\ljStr), & \ljObj=\emptyset ~\wedge~ \ljProto=\ljLocation\\
                                                                \ljUndefined, & \ljObj=\emptyset ~\wedge~ \ljProto=\ljConst
                                                \end{cases}\\

                                                \langle \ljObj, \ljClosure, \ljProto \rangle[\ljStr\mapsto\ljVal] ~&::=&~ \langle \ljObj[\ljStr\mapsto\ljVal], \ljClosure, \ljProto \rangle\\

                                                \langle \ljObj, \ljClosure, \ljProto \rangle_{\ljClosure}  ~&::=&~ \ljClosure\\

                                                \ljHeap[\ljLocation,\ljStr \mapsto \ljVal] ~&::=&~
                                                \ljHeap[\ljLocation\mapsto\ljHeap(\ljLocation)[\ljStr \mapsto \ljVal]]\\

                                                \ljHeap[\ljLocation \mapsto \emptyset] ~&::=&~ \ljHeap[\ljLocation \mapsto \langle \emptyset,\emptyset,\ljNull \rangle]\\

                                                \ljHeap[\ljLocation \mapsto \ljObj] ~&::=&~ \ljHeap[\ljLocation \mapsto \langle \ljObj,\emptyset,\ljNull \rangle]\\

                                                \ljHeap[\ljLocation \mapsto \ljClosure] ~&::=&~ \ljHeap[\ljLocation \mapsto \langle \emptyset,\ljClosure,\ljNull \rangle]\\

                                                \ljHeap[\ljLocation \mapsto \ljProto] ~&::=&~ \ljHeap[\ljLocation \mapsto \langle \emptyset,\emptyset,\ljProto \rangle]

                                \end{array}
                \end{displaymath}
                \caption{Abbreviations.}
                \label{fig:syntactic_sugar}
\end{figure}

Figure \ref{fig:syntactic_sugar} introduces some abbreviated notation. A property lookup or a property update on a
storable $\ljStorable=\langle \ljObj, \ljClosure, \ljProto \rangle$ is relayed to the underlying object. The property access $\ljStorable(\ljStr)$ returns
$\ljUndefined$ by default if the accessed string is not defined in $\ljObj$ and the prototype of $\ljStorable$
is not a location $\ljLocation$. We write $\ljStorable_{\ljClosure}$ for the closure in $\ljStorable$. The notation
$\ljHeap[\ljLocation,\ljStr \mapsto \ljVal]$ updates a property of storable $\ljHeap(\ljLocation)$,
$\ljHeap[\ljLocation \mapsto \ljObj]$ initializes an object, and
$\ljHeap[\ljLocation \mapsto \ljClosure]$ defines a function.


\subsection{Independence} \label{sec:dependency}

The $\ljTrace$ expression serves to mark a program point
as a source of sensitive data. An expression $\ljExp$ is independent
from that source if the value of the $\ljTrace$ expression does not
influence the final result of $\ljExp$. The first definition formalizes replacing
the argument of a $\ljTrace$ expression.

\begin{definition}\label{def:substitution-i}
                The substitution
                $\ljExp[\ljSLocation\mapsto\tilde{\ljExp}]$ of
                $\ljSLocation$ in $\ljExp$ by $\tilde{\ljExp}$ is
                defined as the homomorphic extension of
                \begin{equation}
                                \ljTrace(\ljExp')[\ljSLocation\mapsto\tilde{\ljExp}]\equiv\ljTrace(\tilde{\ljExp})
                \end{equation}
\end{definition}




\begin{definition}[incomplete first attempt]\label{def:independency}
                The expression $\ljExp$ is independent from $\ljSLocation$ iff all possible substitutions of $\ljSLocation$ are unobservable.
                \begin{equation}
                                \begin{split}
                                                \forall  \ljExp_1,\ljExp_2:~ &\ljHeap,\ljEnv ~\entails~\ljExp[\ljSLocation\mapsto\ljExp_1] ~\eval~ \ljHeap_1 ~|~ \ljVal\\
                                                ~\leftrightarrow~ &\ljHeap,\ljEnv ~\entails~\ljExp[\ljSLocation\mapsto\ljExp_2] ~\eval~ \ljHeap_2 ~|~ \ljVal \\
                                \end{split}
                \end{equation}
\end{definition}

This definition covers both, the terminating and the non-terminating cases.
Furthermore, we consider direct dependencies, indirect dependencies, and
transitive dependencies, similar to the behavior described by Denning
\cite{Denning:1976:LMS:360051.360056,Denning:1977:CPS:359636.359712}.
In Section~\ref{sec:technical_results}, we complete this definition to make it amenable to proof.


\section{Dependency Tracking Semantics}
\label{sec:dependency_type}

To attach marker propagation for upcoming values we apply definition \ref{def:independency} to the $\lj$ calculus.
The later on derived abstract interpretation is formalized on this extended calculus.

This section extends the semantics of $\lj$ with mark
propagation. The resulting calculus $\ldj$ \emph{only} provides a
baseline calculus for subsequent static analysis. $\ldj$
is \emph{specifically not} meant to
perform any kind of dynamic analysis, where the presence or absence of
a mark in a value guarantees some dependency related property. 

The calculus extends $\LjValue$ to $\LjTaintedValue \ni~
\ljTypedVal ::= \ljVal:\ljType$ where $\ljType ::= \emptyset ~|~ \ljSource ~|~ \ljType\joinType\ljType$ is a dependency
annotation. $\LjTaintedValue$ replaces $\LjValue$ in objects and environments.  The operation $\joinType$ 
joins two dependencies. If $\ljTypedVal=\ljVal:\ljType_{\ljVal}$ then write $\ljTypedVal\joinType\ljType$ for
$\ljVal:\ljType_{\ljVal}\joinType\ljType$ to apply a dependency annotation to a value. 

\begin{figure}
                \centering
                \begin{mathpar}
                                \inferrule [\LDJConstant]
                                {
                                }
                                {
                                                \ljHeap,\ljEnv,\ljType ~\entails~ \ljConst ~\eval~ \ljHeap ~|~ \ljConst:\ljType
                                }\and
                                \inferrule [\LDJVariable]
                                {
                                }
                                {
                                                \ljHeap,\ljEnv,\ljType ~\entails~ \ljVar ~\eval~ \ljHeap ~|~ \ljEnv(\ljVar)\joinType\ljType
                                }\and
                                \inferrule [\LDJFunctionCreation]
                                {
                                                \ljLocation \notin \dom(\ljHeap)
                                }
                                {
                                                \ljHeap,\ljEnv,\ljType ~\entails~ \ljFunc \ljVar.\ljExp ~\eval~ \ljHeap[\ljLocation\mapsto\langle \ljEnv, \ljFunc \ljVar.\ljExp \rangle] ~|~ \ljLocation:\ljType
                                }\and
                                \inferrule [\LDJOperation]
                                {
                                                \ljHeap,\ljEnv,\ljType ~\entails~  \ljExp_0 ~\eval~ \ljHeap' ~|~ \ljVal_0:\ljType_{0}\\\\
                                                \ljHeap',\ljEnv,\ljType ~\entails~  \ljExp_1 ~\eval~ \ljHeap'' ~|~ \ljVal_1:\ljType_{1} \\\\
                                                \ljVal_{op} ~=~ \ljOperation(\ljVal_0, \ljVal_1)
                                }
                                {
                                                \ljHeap,\ljEnv,\ljType ~\entails~ \ljOp(\ljExp_0, \ljExp_1) ~\eval~ \ljHeap'' ~|~ \ljVal_{op}:\ljType_{0}\joinType\ljType_{1}
                                }\and 
                                \inferrule [\LDJObjectCreation]
                                {
                                                \ljHeap,\ljEnv,\ljType ~\entails~ \ljExp ~\eval~ \ljHeap' ~|~ \ljVal:\ljType_{\ljVal}\\
                                                \ljLocation \notin \dom(\ljHeap)
                                }
                                {
                                                \ljHeap,\ljEnv,\ljType ~\entails~ \ljNew~ \ljExp ~\eval~
                                                \ljHeap'[\ljLocation\mapsto\ljVal] ~|~ \ljLocation:\ljType_{\ljVal}
                                }\and
                                \inferrule [\LDJFunctionApplication]
                                {
                                                \ljHeap,\ljEnv,\ljType ~\entails~ \ljExp_0 ~\eval~ \ljHeap' ~|~ \ljLocation:\ljType_{0} \\
                                                \langle \ljObj, \langle \dot{\ljEnv}, \ljFunc \ljVar.\ljExp \rangle, \ljProto \rangle = \ljHeap'(\ljLocation) \\\\
                                                \ljHeap',\ljEnv,\ljType ~\entails~ \ljExp_1 ~\eval~ \ljHeap'' ~|~ \ljVal_{1}:\ljType_{1} \\\\
                                                \ljHeap'',\dot{\ljEnv}[\ljVar \mapsto \ljVal_{1}:\ljType_{1}],\ljType\joinType\ljType_{0} ~\entails~ \ljExp ~\eval~ \ljHeap''' ~|~ \ljVal:\ljType_{\ljVal}
                                }
                                {
                                                \ljHeap,\ljEnv,\ljType ~\entails~ \ljExp_0(\ljExp_1) ~\eval~ \ljHeap''' ~|~ \ljVal:\ljType_{\ljVal}
                                }\and
                                \inferrule [\LDJConditionTrue]
                                {
                                                \ljHeap,\ljEnv,\ljType ~\entails~ \ljExp_0 ~\eval~ \ljHeap' ~|~ \ljVal_{0}:\ljType_{0} \\\\
                                                \ljVal_0 = \ljTrue \\
                                                \ljHeap',\ljEnv,\ljType\joinType\ljType_{0} ~\entails~ \ljExp_1 ~\eval~ \ljHeap_{1}'' ~|~ \ljVal_{1}:\ljType_{1}
                                }
                                {
                                                \ljHeap,\ljEnv,\ljType ~\entails~ \ljIf~ (\ljExp_0)~ \ljExp_1,~ \ljExp_2 ~\eval~ \ljHeap_{1}'' ~|~ \ljVal_1:\ljType_{1}
                                }\and
                                \inferrule [\LDJConditionFalse]
                                {
                                                \ljHeap,\ljEnv,\ljType ~\entails~ \ljExp_0 ~\eval~ \ljHeap' ~|~ \ljVal_0:\ljType_{0} \\\\
                                                \ljVal_0 \neq \ljTrue \\
                                                \ljHeap',\ljEnv,\ljType \joinType \ljType_{0} ~\entails~ \ljExp_2 ~\eval~ \ljHeap_{2}'' ~|~ \ljVal_2:\ljType_{2}
                                }
                                {
                                                \ljHeap,\ljEnv,\ljType ~\entails~ \ljIf~ (\ljExp_0)~ \ljExp_1,~ \ljExp_2 ~\eval~ \ljHeap_{2}'' ~|~ \ljVal_2:\ljType_{2}
                                }\and
                                \inferrule [\LDJPropertyReference]
                                {
                                                \ljHeap,\ljEnv,\ljType ~\entails~ \ljExp_{0} ~\eval~ \ljHeap' ~|~ \ljLocation:\ljType_{\ljLocation} \\\\
                                                \ljHeap',\ljEnv,\ljType ~\entails~ \ljExp_{1} ~\eval~ \ljHeap'' ~|~ \ljStr:\ljType_{\ljStr}
                                }
                                {
                                                \ljHeap,\ljEnv,\ljType ~\entails~ \ljExp_{0}[\ljExp_{1}] ~\eval~ \ljHeap'' ~|~
                                                \ljHeap''(\ljLocation)(\ljStr)\joinType\ljType_{\ljLocation}\joinType\ljType_{\ljStr}
                                }\and
                                \inferrule [\LDJPropertyAssignment]
                                {
                                                \ljHeap,\ljEnv,\ljType ~\entails~ \ljExp_{0} ~\eval~ \ljHeap' ~|~ \ljLocation:\ljType_{\ljLocation} \\\\
                                                \ljHeap',\ljEnv,\ljType ~\entails~ \ljExp_{1} ~\eval~ \ljHeap'' ~|~ \ljStr:\ljType_{\ljStr} \\\\
                                                \ljHeap'',\ljEnv,\ljType ~\entails~ \ljExp_{2} ~\eval~ \ljHeap''' ~|~ \ljVal:\ljType_{\ljVal}\\
                                                \ljHeap'''' = \ljHeap'''[\ljLocation,\ljStr \mapsto \ljVal:\ljType_{\ljVal} \joinType \ljType_{\ljLocation} \joinType \ljType_{\ljStr}]
                                }
                                {
                                                \ljHeap,\ljEnv,\ljType ~\entails~ \ljExp_0[\ljExp_1] = \ljExp_2 ~\eval~ \ljHeap'''' ~|~ \ljVal:\ljType_{\ljVal}
                                }\and
                                \inferrule [\LDJTrace]
                                {
                                                \ljHeap,\ljEnv,\ljType\joinType\ljSource ~\entails~ \ljExp ~\eval~ \ljHeap' ~|~ \ljVal:\ljType_{\ljVal}
                                }
                                {
                                                \ljHeap,\ljEnv,\ljType ~\entails~ \ljTrace~ (\ljExp) ~\eval~ \ljHeap' ~|~ \ljVal:\ljType_{\ljVal}
                                }
                \end{mathpar}
                \caption{Inference rules of $\ldj$.}
                \label{fig:inference-rules_ldj}
\end{figure}

The big-step evaluation judgment  $\ljHeap,\ljEnv,\ljType ~\entails~ \ljExp ~\eval~
\ljHeap' ~|~ \ljTypedVal$ for $\ldj$ extends the one for $\lj$ by a
new component $\ljType$ which tracks the context dependency for expression $\ljExp$.  Figure
\ref{fig:inference-rules_ldj} contains its defining inference rules.

The evaluation rules \Rule{\LDJConstant}, \Rule{\LDJVariable}, and \Rule{\LDJFunctionCreation} are trivial. 
Their return values depend on the context.  \Rule{\LDJOperation} calculates the
result on the value part and combines the dependencies of the involved
values. $\ljOperation$ stands for the application of operator $\ljOp$.
The rule \Rule{\LDJObjectCreation} binds the 
dependency of the evaluated prototype to the returned location.  During \Rule{\LDJFunctionApplication}
the dependency of the value referencing the function is bound to the sub-context.  In a similar way
\Rule{\LDJConditionTrue} and \Rule{\LDJConditionFalse} bind the dependency of the condition to the sub-context.
The rule \Rule{\LDJPropertyReference} combines the dependencies of heap location and property reference to the returned value.
The rule \Rule{\LDJPropertyAssignment} combines these dependencies to the assigned value because the evaluated location and
property references affect the write operation and further the value which is accessible at this location.

The trace expression $\ljTrace~ (\ljExp)$ \Rule{\LDJTrace} adds the
$\ljSource$ annotation to the context of expression 
$\ljExp$. This addition causes all values created or modified in
$\ljExp$ to be marked with $\ljSource$ (e.g. to detect side effects) as stated by the following context
dependency lemma.

\begin{lemma}
                \label{thm:context_dependency}
                $\ljHeap,\ljEnv,\ljType \entails \ljExp \eval \ljHeap' \mid \ljVal:\ljType_{\ljVal}$
                implies that
                $\ljType\subseteq\ljType_{\ljVal}$.
\end{lemma}

The proof is by induction on the relation $\eval$ (Section \ref{sec:proof_context}).

\section{Abstract Analysis} \label{sec:abstract_analysis}

\begin{figure}[t]
                \centering
                \begin{displaymath}
                                \begin{array}{llrll}
                                                \LvUndefined &&::=& \powerset(\{\lvUndefined\})\\
                                                \LvNull &&::=& \powerset(\{\lvNull\})\\
                                                \LvBool &&::=& \powerset(\{\lvTrue,\lvFalse\})\\
                                                \LvNum &&::=& \lvNUM^\top_\perp\\
                                                \LvString &&::=& \lvCHAR^\top_\perp\\
                                                \ALatticeValue &\ni~ \aLattice&::=& \LvUndefined ~\times~ \LvNull ~\times~\\
                                                &&&\LvBool ~\times~ \LvNum ~\times~ \LvString
                                \end{array}
                \end{displaymath}
                \caption{Base Type Value Lattice.}
                \label{fig:lattice-value}
\end{figure}

\begin{figure}[t]
                \centering
                \begin{displaymath}
                                \begin{array}{llrll}
                                                \ALabel &\ni~ \aObjlabel &::=& \{\ljSLocation\dots\}\\
                                                \AClosure &\ni~  \aFunc &::=& \langle \aScope, \ljFunc \ljVar.\ljExp \rangle\\
                                                \AObject &\ni~ \aObj &::=& \emptyset ~|~ \aObj[\aLattice \mapsto \aVal]\\
                                                \AValue &\ni~ \aVal &::=& \langle \aLattice, \aObjlabel, \D \rangle\\
                                                \AStorable &\ni~ \aStorable &::=& \langle \aObj, \aFunc, \aObjlabel \rangle\\
                                                \AFunctionStore &\ni~ \aStore &::=& \emptyset ~|~ \aStore[\ljSLocation \mapsto \langle \aState, \aVal, \aState, \aVal \rangle]\\
                                                \AScope &\ni~ \aScope &::=& \emptyset ~|~ \aScope[\ljVar \mapsto \aVal]\\
                                                \AStorableStore &\ni~ \aObjStore &::=& \emptyset ~|~ \aObjStore[\ljSLocation \mapsto \aStorable]\\
                                                \AState &\ni~ \aState &::=& \langle \aObjStore, \D \rangle\\
                                                \DDependency &\ni~ \D &::=& \emptyset ~|~ \ljSLocation ~|~ \D \join \D\\
                                \end{array}
                \end{displaymath}
                \caption{Abstract Semantic Domains.}
                \label{fig:abstract-semantic-domains}
\end{figure}

The analysis is an abstraction of the $\ldj$ calculus. Its basis is the lattice for base type values
(Figure \ref{fig:lattice-value}), which is a simplified adaptation of
the lattice of TAJS  \cite{tajs2009}.  $\lvNUM$ is the set of floating point numbers, $\lvCHAR$ the set of string literals,
and the annotation ${\cdot}^\top_\bot$ turns a set into a flat lattice by adding a bottom and top element. An element of the analysis
lattice is a tuple like $\langle\perp,\perp,\lvTrue,\perp,\mathtt{"x"}\rangle$ which represents a value which is either
the boolean value $\lvTrue$ or the string $\mathtt{"x"}$. Further, $\langle\perp,\perp,\perp,\top,\perp\rangle$
represents all possible number values. The abstract semantic domains
(Figure \ref{fig:abstract-semantic-domains}) are similar to the domains arising from the $\ldj$ calculus except that 
a set of marks $\aObjlabel$ abstracts a set of concrete locations $\ljLocation$ where $\ljSLocation\in\aObjlabel$. 
An abstract value $\aVal ~=~ \langle \aLattice,\aObjlabel,\D \rangle$ is a triple of a lattice element $\aLattice$,
object marks $\aObjlabel$, and dependency $\D$.

Hence, each abstract value represents a set of base type values and a set of objects.  We write $\aLattice_{\aVal}$ for
the analysis lattice, $\aObjlabel_{\aVal}$ for the marks, and $\D_{\aVal}$ for the dependency component of the
abstract value $\aVal$.  Each abstract object is identified by the
mark $\ljSLocation$ corresponding to the $\ljNew~\ljExp$ 
expression creating the object. 
An abstract storable consists of an abstract object, a function closure, and a set of locations representing the prototype.
Unlike before, the abstract object maps a lattice element to a value.
This mechanism reduces the number of merge operations during the abstract analysis.

The abstract state is a pair $\aState = \langle \aObjStore,\D \rangle$ where $\aObjStore$ is the mapping from
marks to abstract storables. We write $\aObjStore_{\aState}$ for the object store, and $\D_{\aState}$ for the dependency in
$\aState$. $\aState(\aObjlabel)$ provides a set of storables, denoted
by $\aStorables$. The substitution of $\D$ in $\aState$  
written $\aState[\D\mapsto\D']\equiv\langle\aObjStore_{\aState},\D'\rangle$ replaces the state dependency.

To handle recursive function calls we introduce a global function store $\aStore$, which maps a mark $\ljSLocation$
to two pairs of state $\aState$ and value $\aVal$.  Functions are also
identified by marks $\ljSLocation$.  The function store contains  the
merged result of the last evaluation for each function. The first pair $\aState_{\aFuncInput},
\aVal_{\aFuncInput}$ represents the input state and input parameter of all heretofore taken function calls, the second
one $\aState_{\aFuncOutput}, \aVal_{\aFuncOutput}$ the output state and return value. For further use we write
$\aStore(\ljSLocation)_{\aFuncInput}$ to select the input, and $\aStore(\ljSLocation)_{\aFuncOutput}$ for the output.
The substitutions $\aStore[\ljSLocation,\aFuncInput\mapsto\langle\aState,\aVal\rangle]$ and
$\aStore[\ljSLocation,\aFuncOutput\mapsto\langle\aState,\aVal\rangle]$ denotes the store update operation
on input or output pairs.

Its inference is stated by the following lemma.

\begin{lemma}[Function Store] \label{thm:function_store}
  $\forall\aStore,\aFunc,\aState,\aVal:$
  If $\langle\aState,\aVal\rangle\sqsubseteq\aStore(\ljSLocation)_{\aFuncInput}$
  and $\aFunc = \langle \dot{\aScope}, \ljFunc\ljVar.\ljExp \rangle$
  then $\aState,\dot{\aScope}[\ljVar\mapsto\aVal] ~\entails~ \ljExp ~\eval~ \aState' ~|~ \aVal'$
  and $\langle\aState',\aVal'\rangle\sqsubseteq\aStore(\ljSLocation)_{\aFuncOutput}$.
\end{lemma}

                                The proof is by induction on the
                derivation of $\aState''[\D \mapsto \D_{\aState''} \join \D_{0}] ~\entailsFAIteration~ \aState''(\aObjlabel_{0}),\aVal_{1} ~\eval~ \aState''' ~|~ \aVal$.

The $\ljTrace$ expression registers the dependency from $\ell$ on all
values that pass through it.

The abstraction is defined as relation between $\ljVal\in\LjTaintedValue$ and $\aVal\in\AValue$.

\begin{definition}[Abstraction] \label{def:abstraction_function}
                The abstraction $\alpha:\LjTaintedValue \\ \rightarrow \AValue$ is defined as:
                \begin{align}
                                \alpha(\ljVal:\ljType) ~::=~
                                \begin{cases}
                                                \langle \perp,\{\ljSLocation\},\{\ljSource|\ljSource\in\ljType\} \rangle & \ljVal=\ljLocation\\
                                                \langle \ljConst,\emptyset,\{\ljSource|\ljSource\in\ljType\} \rangle & \ljVal=\ljConst
                                \end{cases}
                \end{align}
\end{definition}

\begin{definition}[Abstract Operation] \label{def:abstract_operation}
                The abstract operation $\abstractOperation$ is defined
                in terms of the concrete operation $\ljOperation$ as usual:
                \begin{align}
                                \begin{split}
                                                &\abstractOperation(\aVal_{0},\aVal_{1}) ~::=~ \\
                                                &\bigsqcup \{ \ljOperation(\ljVal_{0},\ljVal_{1}) ~|~ \ljVal_{0}\in\aVal_{0}, \ljVal_{1}\in\aVal_{1} \}
                                \end{split}
                \end{align}
                                This definition implies that:
                \begin{align}
                                \begin{split}
                                                &\ljOperation(\ljVal_{0},\ljVal_{1})=\ljVal_{op} ~\rightarrow~\\
                                                &\abstractOperation(\alpha(\ljVal_{0}),\alpha(\ljVal_{1}))\sqsupseteq\alpha(\ljVal_{op})
                                \end{split}
                \end{align}
\end{definition}


\begin{figure}[t]
                \centering
                \begin{mathpar}
                                \inferrule [\AProgram]
                                {
                                                \aStateBottom,\aScopeBottom ~\entails~ \ljExp ~\eval~ \aState ~|~ \aVal\\\\
                                                \entailsP~ \langle \aStoreBottom, \aStateBottom, \aValBottom \rangle, \langle \aStore, \aState, \aVal \rangle, \ljExp ~\eval~ \aState' ~|~ \aVal'
                                }
                                {
                                                \entails~ \ljExp ~\eval~ \aState' ~|~ \aVal'
                                }\and
                                \inferrule [\AProgramIterationNotEquals]
                                {
                                                \aStateBottom,\aScopeBottom ~\entails~ \ljExp ~\eval~ \aState' ~|~ \aVal' \\\\
                                                \entailsP~ \aALatticeR',\langle \aStore, \aState', \aVal' \rangle, \ljExp ~\eval~ \aALatticeQ
                                }
                                {
                                                \entailsP~ \aALatticeR,\aALatticeR',\ljExp ~\eval~ \aALatticeQ
                                }\and
                                \inferrule [\AProgramIterationEquals]
                                {
                                }
                                {
                                                \entailsP~ \aALatticeR,\aALatticeR,\ljExp ~\eval~ \aALatticeR
                                }
                \end{mathpar}
                \caption{Inference rules for program interpretation.}
                \label{fig:inference-rules_program}
\end{figure}


\begin{figure*}[!htb]
                \centering
                \begin{mathpar}
                                \inferrule [\AConstant]
                                {
                                }
                                {
                                                \aState,\aScope ~\entails~ \ljConst ~\eval~ \aState ~|~ \langle \ljConst, \emptyset, \D_{\aState} \rangle
                                }\and
                                \inferrule [\AVariable]
                                {
                                }
                                {
                                                \aState,\aScope ~\entails~ \ljVar ~\eval~ \aState ~|~ \aScope(\ljVar) \join \D_{\aState}
                                }\and
                                \inferrule [\AOperation]
                                {
                                                \aState,\aScope ~\entails~ \ljExp_0 ~\eval~ \aState' ~|~ \aVal_{0}\\\\
                                                \aState',\aScope ~\entails~ \ljExp_1 ~\eval~ \aState'' ~|~ \aVal_{1} \\\\
                                                \langle\aLattice, \aObjlabel\rangle = \abstractOperation(\aVal_0, \aVal_1)
                                }
                                {
                                                \aState,\aScope ~\entails~ \ljOp(\ljExp_0, \ljExp_1) ~\eval~ \aState'' ~|~ \langle\aLattice,\aObjlabel,\D_{\aVal_0}\sqcup\D_{\aVal_1}\rangle
                                }\and
                                \inferrule [\AObjectCreation]
                                {
                                                \ljSLocation \notin \dom(\aState) \\
                                                \aState,\aScope ~\entails~ \ljExp ~\eval~ \aState' ~|~ \langle \aLattice, \aObjlabel, \D \rangle
                                }
                                {
                                                \aState,\aScope ~\entails~ \ljNew~ \ljExp
                                                ~\eval~ \aState'[\ljSLocation\mapsto\aObjlabel] ~|~ \langle\aLatticeBottom,\{\ljSLocation\},\D_{\aState}\join\D\rangle
                                }\and
                                \inferrule [\AObjectCreationExisting]
                                {
                                                \ljSLocation \in \dom(\aState) \\
                                                \aState,\aScope ~\entails~ \ljExp ~\eval~ \aState' ~|~ \langle \aLattice, \aObjlabel, \D \rangle
                                }
                                {
                                                \aState,\aScope ~\entails~ \ljNew~ \ljExp
                                                ~\eval~ \aState'[\ljSLocation\mapsto\aState(\ljSLocation)\join\langle\emptyset,\aFuncBottom,\aObjlabel\rangle]
                                                ~|~ \langle\aLatticeBottom,\{\ljSLocation\},\D_{\aState}\join\D\rangle
                                }\and
                                \inferrule [\AFunctionCreation]
                                {
                                                \ljSLocation \notin \dom(\aState) \\
                                                \aStore[\ljSLocation\mapsto\langle \aStateBottom, \aValBottom, \aStateBottom, \aValBottom \rangle]
                                }
                                {
                                                \aState,\aScope ~\entails~ \ljFunc \ljVar.\ljExp ~\eval~ \aState[\ljSLocation\mapsto\langle \aScope, \ljFunc \ljVar.\ljExp \rangle]  ~|~
                                                \langle\aLatticeBottom,\{\ljSLocation\},\D_{\aState}\rangle
                                }\and
                                \inferrule [\AFunctionCreationExisting]
                                {
                                                \ljSLocation \in \dom(\aState) \\
                                                \langle \dot{\aScope}, \ljFunc\ljVar.\ljExp \rangle = \aState(\ljSLocation)_{\aFunc}
                                }
                                {
                                                \aState,\aScope ~\entails~ \ljFunc \ljVar.\ljExp ~\eval~ \aState[\ljSLocation\mapsto\langle \aScope\join\dot{\aScope}, \ljFunc \ljVar.\ljExp \rangle] ~|~
                                                \langle\aLatticeBottom,\{\ljSLocation\},\D_{\aState}\rangle
                                }\and
                                \inferrule [\AFunctionApplication]
                                {
                                                \aState,\aScope ~\entails~ \ljExp_0 ~\eval~ \aState' ~|~ \langle \aLattice_{0}, \aObjlabel_{0}, \D_{0} \rangle\\\\
                                                \aState', \aScope ~\entails~ \ljExp_1 ~\eval~ \aState'' ~|~ \aVal_{1}\\\\
                                                \aState''[\D \mapsto \D_{\aState''} \join \D_{0}] ~\entailsFAIteration~ \aState''(\aObjlabel_{0}),\aVal_{1} ~\eval~ \aState''' ~|~ \aVal
                                }
                                {
                                                \aState,\aScope ~\entails~ \ljExp_0(\ljExp_1) ~\eval~ \langle\aObjStore_{\aState'''},\D_{\aState}\rangle ~|~ \aVal
                                }\and
                                \inferrule [\APropertyReference]
                                {
                                                \aState,\aScope ~\entails~ \ljExp_0 ~\eval~ \aState' ~|~ \langle \aLattice_{0}, \aObjlabel_{0}, \D_{0} \rangle \\\\
                                                \aState',\aScope ~\entails~ \ljExp_1 ~\eval~ \aState'' ~|~ \langle \lvString, \aObjlabel_{1}, \D_{1} \rangle \\\\
                                                \aState'' ~\entailsPRIteration~ \aState''(\aObjlabel_{0}),\lvString ~\eval~ \aVal
                                }
                                {
                                                \aState,\aScope ~\entails~ \ljExp_0[\ljExp_1] ~\eval~ \aState'' ~|~ \langle \aLattice_{\aVal}, \aObjlabel_{\aVal}, \D_{0} \join \D_{1} \join \D_{\aVal} \rangle
                                }\and
                                \inferrule [\APropertyAssignment]
                                {
                                                \aState,\aScope ~\entails~ \ljExp_0 ~\eval~ \aState' ~|~ \langle \aLattice_{0}, \aObjlabel_{0}, \D_{0} \rangle \\\\
                                                \aState',\aScope ~\entails~ \ljExp_1 ~\eval~ \aState'' ~|~ \langle \lvString, \aObjlabel_{1}, \D_{1} \rangle \\\\
                                                \aState'',\aScope ~\entails~ \ljExp_2 ~\eval~ \aState''' ~|~ \aVal \\\\
                                                \aState''' ~\entailsPAIteration~ \aObjlabel_{0},\lvString,\langle \aLattice_{\aVal}, \aObjlabel_{\aVal},\D_{0} \join \D_{1} \join \D_{\aVal} \rangle ~\eval~ \aState''''
                                }
                                {
                                                \aState,\aScope ~\entails~ \ljExp_0[\ljExp_1] = \ljExp_2 ~\eval~ \aState'''' ~|~ \aVal
                                }\and
                                \inferrule [\AConditionTrue]
                                {
                                                \aState,\aScope ~\entails~ \ljExp_0 ~\eval~ \aState' ~|~ \langle \aLattice_{0}, \aObjlabel_{0}, \D_{0} \rangle \\\\
                                                \aLattice_{0}= \lvTrue \\
                                                \aState'[\D \mapsto \D_{\aState'} \join \D_{0}],\aScope ~\entails~ \ljExp_1 ~\eval~ \aState'' ~|~ \aVal_1
                                }
                                {
                                                \aState,\aScope ~\entails~ \ljIf~ (\ljExp_0)~ \ljExp_1,~ \ljExp_2 ~\eval~ \langle\aObjStore_{\aState''},\D_{\aState}\rangle ~|~ \aVal_1
                                }\and
                                \inferrule [\AConditionFalse]
                                {
                                                \aState,\aScope ~\entails~ \ljExp_0 ~\eval~ \aState' ~|~ \langle \aLattice_{0}, \aObjlabel_{0}, \D_{0} \rangle \\\\
                                                \aLattice_{0}= \lvFalse \\
                                                \aState'[\D \mapsto \D_{\aState'} \join \D_{0}],\aScope ~\entails~ \ljExp_2 ~\eval~ \aState'' ~|~ \aVal_2
                                }
                                {
                                                \aState,\aScope ~\entails~ \ljIf~ (\ljExp_0)~ \ljExp_1,~ \ljExp_2 ~\eval~ \langle\aObjStore_{\aState''},\D_{\aState}\rangle ~|~ \aVal_2
                                }\and
                                \inferrule [\ACondition]
                                {
                                                \aState,\aScope ~\entails~ \ljExp_0 ~\eval~ \aState' ~|~ \langle \aLattice_{0}, \aObjlabel_{0}, \D_{0} \rangle \\\\
                                                \aLattice_0 \neq \lvTrue \wedge \aLattice_0 \neq \lvFalse  \\\\
                                                \aState'[\D \mapsto \D_{\aState'} \join \D_{0}],\aScope ~\entails~ \ljExp_1 ~\eval~ \aState_1'' ~|~ \aVal_1 \\\\
                                                \aState'[\D \mapsto \D_{\aState'} \join \D_{0}],\aScope ~\entails~ \ljExp_2 ~\eval~ \aState_2'' ~|~ \aVal_2
                                }
                                {
                                                \aState,\aScope ~\entails~ \ljIf~ (\ljExp_0)~ \ljExp_1,~ \ljExp_2 ~\eval~ \langle\aObjStore_{\aState_1''}\join\aObjStore_{\aState_2''},\D_{\aState}\rangle
                                                ~|~ \aVal_1 \join \aVal_2
                                }\and
                                \inferrule [\ATrace]
                                {
                                                \aState[\D\mapsto\D_{\aState}\join\ljSource], \aScope ~\entails~ \ljExp ~\eval~ \aState' ~|~ \aVal
                                }
                                {
                                                \aState,\aScope ~\entails~ \ljTrace~ (\ljExp) ~\eval~ \langle \aObjStore_{\aState'}, \D_{\aState} \rangle ~|~ \aVal
                                }
                \end{mathpar}
                \caption{Inference rules for abstract interpretation.}
                \label{fig:inference-rules_abstract}
\end{figure*}

\begin{figure}[t]
                \centering
                \begin{mathpar}
                                \inferrule [\AFunctionIteration]
                                {
                                                \aState ~\entailsFAApplication~ \aFunc,\aVal ~\eval~ \aState' ~|~ \aVal'\\\\
                                                \aState' ~\entailsFAIteration~ \aStorables,\aVal ~\eval~ \aState'' ~|~ \aVal''
                                }
                                {
                                                \aState ~\entailsFAIteration~ \langle\aObj,\aFunc,\aObjlabel\rangle;\aStorables,\aVal ~\eval~ \aState'' ~|~ \aVal' \join \aVal''
                                }\and
                                \inferrule [\AFunctionIterationEmpty]
                                {
                                }
                                {
                                                \aState ~\entailsFAIteration~ \emptyset,\aVal ~\eval~ \aState ~|~ \aValBottom
                                }\and
                                \inferrule [\AFunctionStoreSubset]
                                {
                                                \langle\aState,\aVal\rangle \sqsubseteq \aStore(\ljSLocation)_{\aFuncInput}\\
                                                \langle\aState',\aVal'\rangle = \aStore(\ljSLocation)_{\aFuncOutput}
                                }
                                {
                                                \aState ~\entailsFAApplication~ \aFunc,\aVal ~\eval~ \aState' ~|~ \aVal'
                                }\and
                                \inferrule [\AFunctionStoreNonSubset]
                                {
                                                \langle\aState,\aVal\rangle \not\sqsubseteq \aStore(\ljSLocation)_{\aFuncInput}\\
                                                \langle \dot{\aScope}, \ljFunc\ljVar.\ljExp \rangle = \aFunc\\\\
                                                \langle\bar{\aState},\bar{\aVal}\rangle = \aStore(\ljSLocation)_{\aFuncInput} \join \langle \aState,\aVal\rangle\\
                                                \aStore[\ljSLocation,\aFuncInput\mapsto\langle\bar{\aState},\bar{\aVal}\rangle]\\\\
                                                \bar{\aState},\dot{\aScope}[\ljVar \mapsto \bar{\aVal}] ~\entails~ \ljExp ~\eval~ \bar{\aState}' ~|~ \bar{\aVal}'\\
                                                \aStore[\ljSLocation,\aFuncOutput\mapsto\langle\bar{\aState}',\bar{\aVal}' \rangle]
                                }
                                {
                                                \aState ~\entailsFAApplication~ \aFunc,\aVal ~\eval~ \bar{\aState}' ~|~ \bar{\aVal}'
                                }       
                \end{mathpar}
                \caption{Inference rules for function application.}
                \label{fig:inference-rules_function-application}
\end{figure}

\begin{figure}[t]
                \centering
                \begin{mathpar}
                                \inferrule [\APropertyReferenceIteration]
                                {
                                                \aState ~\entailsPRIntersection~ \aObj,\aLattice ~\eval~ \aVal\\\\
                                                \aState ~\entailsPRIteration~ \aStorables,\aLattice ~\eval~ \aVal'\\\\
                                                \aState ~\entailsPRIteration~ \aState(\aObjlabel),\aLattice ~\eval~ \aVal''
                                }
                                {
                                                \aState ~\entailsPRIteration~ \langle\aObj,\aFunc,\aObjlabel\rangle;\aStorables,\aLattice ~\eval~ \aVal\join\aVal'\join\aVal''
                                }\and
                                \inferrule [\APropertyReferenceIterationEmpty]
                                {
                                }
                                {
                                                \aState ~\entailsPRIteration~ \emptyset,\aLattice ~\eval~ \aValBottom
                                }\and
                                \inferrule [\APropertyReferenceIntersection]
                                {
                                                \aLattice\sqcap\aLattice_{i}\neq\perp\\\\
                                                \aState ~\entailsPRIteration~ \aObj,\aLattice ~\eval~ \aVal'
                                }
                                {
                                                \aState ~\entailsPRIntersection~ (\aLattice_{i}:\aVal_{i});\aObj, \aLattice ~\eval~ \aVal_{i}\join\aVal'
                                }\and
                                \inferrule [\APropertyReferenceNonIntersection]
                                {
                                                \aLattice\sqcap\aLattice_{i}=\perp\\\\
                                                \aState ~\entailsPRIteration~ \aObj,\aLattice ~\eval~ \aVal'
                                }
                                {
                                                \aState ~\entailsPRIntersection~ (\aLattice_{i}:\aVal_{i});\aObj,\aLattice ~\eval~ \aVal'
                                }\and
                                \inferrule [\APropertyReferenceEmpty]
                                {
                                }
                                {
                                                \aState ~\entailsPRIntersection~ \emptyset, \aLattice ~\eval~ \langle \langle \top,\perp,\perp,\perp,\perp \rangle, \emptyset, \emptyset \rangle
                                }
                \end{mathpar}
                \caption{Inference rules for property reference.}
                \label{fig:inference-rules_property-reference}
\end{figure}

\begin{figure}[t]
                \centering
                \begin{mathpar}
                                \inferrule [\APropertyAssignmentIteration]
                                {
                                                \aState ~\entailsPAAssignment~ \ljSLocation,\aLattice,\aVal ~\eval~ \aState'\\\\
                                                \aState' ~\entailsPAIteration~ \aObjlabel,\aLattice,\aVal ~\eval~ \aState''
                                }
                                {
                                                \aState ~\entailsPAIteration~ \ljSLocation;\aObjlabel,\aLattice,\aVal ~\eval~ \aState''
                                }\and
                                \inferrule [\APropertyAssignmentIterationEmpty]
                                {
                                }
                                {
                                                \aState ~\entailsPAIteration~ \emptyset,\aLattice,\aVal ~\eval~ \aState
                                }\and
                                \inferrule [\APropertyAssignmentInDom]
                                {
                                                \aLattice \in \dom(\aState(\ljSLocation))
                                }
                                {
                                                \aState ~\entailsPAAssignment~ \ljSLocation,\aLattice,\aVal ~\eval~ \aState[\ljSLocation,\aLattice \mapsto \aState(\ljSLocation)(\aLattice) \join \aVal]
                                }\and
                                \inferrule [\APropertyAssignmentNotInDom]
                                {
                                                \aLattice \notin \dom(\aState(\ljSLocation))
                                }
                                {
                                                \aState ~\entailsPAAssignment~ \ljSLocation,\aLattice,\aVal ~\eval~ \aState[\ljSLocation,\aLattice \mapsto \aVal]
                                }
                \end{mathpar}
                \caption{Inference rules for property assignment.}
                \label{fig:inference-rules_property-assignment}
\end{figure}

Figures \ref{fig:inference-rules_program}, \ref{fig:inference-rules_abstract}, 
\ref{fig:inference-rules_function-application}, \ref{fig:inference-rules_property-reference},
and~\ref{fig:inference-rules_property-assignment} show the inference rules for the big-step evaluation judgment of the
abstract semantics.  It has the form $\aState,\aScope ~\entails~ \ljExp ~\eval~ \aState' ~|~ \aVal$.  State $\aState$
and scope $\aScope$ analyze expression $\ljExp$ and result in state $\aState'$ and value $\aVal$.
We use notations similar to Figure~\ref{fig:syntactic_sugar}. 



The global program rule \Rule{\AProgram}
(Figure~\ref{fig:inference-rules_program}) relies on two auxiliary rules
to repeatedly evaluate the 
program until the analysis state, an element of $\AAnalysisLattice$
consisting of $\aStore$, $\aState$ and $\aVal$, becomes stable. In the
figure, 
$\aALatticeR,\aALatticeQ$ range over
$\AAnalysisLattice$ and write
$\aStoreBottom$, $\aStateBottom$ and $\aScopeBottom$ for 
the empty instances of the components.

In Figure~\ref{fig:inference-rules_abstract},
the rules for constants \Rule{\AConstant} and variables
\Rule{\AVariable} work similarly as in $\ldj$. 

The object and function creation rules are also omitted. They check if an object or function, referenced by $\ljSLocation$,
already exists. In this case the object or function creation has to merge the prototypes or scopes.

The rule \Rule{\AOperation} is also standard.  As in $\ldj$ the $\ljTrace$ expression \Rule{\ATrace} assigns mark
$\ljSLocation$ to the sub-state.  The rules for the conditional \Rule{\ACondition}, \Rule{\AConditionTrue},
and \Rule{\AConditionFalse} have to handle the case that it is not possible to
distinguish between $\ljTrue$ and $\ljFalse$.  In this case both branches have to be evaluated and the results merged.

Similar problems arise in function application, property reference, and property assignment. Each value can refer to a
set of objects including a set of prototypes. Therefore each referenced function has to be evaluated \Rule{\AFunctionApplication}
and a property has to be read from \Rule{\APropertyReference} or written to \Rule{\APropertyAssignment} all objects. Results have to be
merged. The auxiliary rules are shown in figure \ref{fig:inference-rules_function-application}, \ref{fig:inference-rules_property-reference},
and~\ref{fig:inference-rules_property-assignment}.

The rules \Rule{\AFunctionIteration} and \Rule{\AFunctionIterationEmpty} iterate over all referenced functions.
Function application relies on the function store $\aStore$.  Before evaluating the function body, the
analyzer checks if the input, consisting of $\aState$ and parameter $\aVal$, is already subsumed by the stored input.
In that case \Rule{\AFunctionStoreSubset}, the stored result, consisting of output state $\aState$ and return value $\aVal$, is used. Otherwise the function body is evaluated \Rule{\AFunctionStoreNonSubset} and the store is updated with the result.


For read and write operations the rules \Rule{\APropertyReferenceIteration}, \Rule{\APropertyReferenceIterationEmpty}, \Rule{\APropertyAssignmentIteration} and \Rule{\APropertyAssignmentIterationEmpty} iterate in a similar way over all references. An abstract object maps a lattice element to a value in case a reference is not a singleton value. All entries having an intersection with the reference are affected by the read operation. The prototype-set has to be involved. \Rule{\APropertyReferenceIntersection}, \Rule{\APropertyReferenceNonIntersection} and \Rule{\APropertyReferenceEmpty} shows its inference. 
Before writing a property, the analyser checks if the property already exists. In this case \Rule{\APropertyAssignmentInDom}, the values get merged. Otherwise \Rule{\APropertyAssignmentNotInDom} the value gets assigned. The actual implementation uses a more refined lattice to improve precision. 

The abstract interpretation over-approximates the dependencies. The
merging of results in  \Rule{\ACondition}, \Rule{\AFunctionApplication}, \Rule{\APropertyReference}, and
\Rule{\APropertyAssignment} may cause false positives.
While some marked values may be independent from the
mark's source, unmarked values are guaranteed to be independent.


\section{Applying the Analysis}
\label{sec:applying-analysis}

This section reconsiders the examples Sensitive Data
(Section~\ref{sec:application_sensitive_data}) and Foreign Code
(Section~\ref{sec:application_foreign_code}) from the introduction
from an abstract analysis point of view.


\subsection{Application: Sensitive Data}
\label{sec:application_sensitive_data_cont}

Given the newly created mark $\ljSLocation_{1}$, the function
\lstinline{userHandler} is initially called with 
$\langle\langle\perp,\perp,\perp,\perp,\texttt{uid1}\rangle,\emptyset,\ljSLocation_{1}\rangle$.
If the result of calling \lstinline{Cookie.isset}
can be determined to be $\lvFalse$, then the dependencies associated with $\lvFalse$ ($\ljSLocation_{1}$ and
$\ljSLocation_{c}$ --- resulting from the cookie interface) are bound to the conditional's context.

The Ajax request cannot be evaluated. So, \lstinline{response} in
\lstinline{onSuccess} is a value containing the
location of an unspecified object like $\emptyset[\langle\perp,\perp,\perp,\perp,\top\rangle\mapsto\langle\langle\perp,\perp,\perp,\perp,\top\rangle,\emptyset,\ljSLocation_{a}\rangle]$
augmented with $\ljSLocation_{a}$. In this case, all further calls to
\lstinline{onSuccess} are already covered by the first input. 

By calling the \lstinline{userHandler} with \lstinline{"uid2"} a new mark $\ljSLocation_{2}$ is introduced.
This call is not covered by the first one so that  the function is reanalyzed with the merged value
$\langle \langle\perp,\perp,\perp,\perp,\LvString\rangle,\emptyset,\{\ljSLocation_{1},\ljSLocation_{2}\}\rangle$.
After  the analysis has stabilized, \lstinline{name1} also depends on $\ljSLocation_{2}$.

The example illustrates that merging functions can result in
conservative results. The implementation has a more refined function
store which is indexed by a pair of 
scope $\aScope$ and source location $\ljSLocation$ to prevent such inaccuracies.


\subsection{Application: Foreign Code}
\label{sec:application_foreign_code_cont}

The \lstinline{trace} expression in line 1
(Section\ref{sec:application_foreign_code}) marks the sub-context for
creating the \lstinline{foreach} function. 
The resulting location that points to the function is augmented with
this mark.
By calling \lstinline{loadForeigenCode} the mark is bound to the
callees context and finally to the value 
referencing the \lstinline{foreach} function. 

By iterating over the array elements (line 11) the dependency
annotation is forwarded to the value occurring in \lstinline{result}.

Unlike many other security analyses, the objects \lstinline{Array} and
\lstinline{Array.prototype} do not receive marks.
If the analysis can determine the updated property exactly, as is the
case with \lstinline{foreach}, then no other properties can be affected
by the update (expect the length). Such an abstract update occurs if
the property name is independent from the input. Otherwise, the update
happens on a approximated set of property names, all of which are
marked by this update.


%


\subsection{Further sample applications}
\label{sec:sample_applications}

We also applied our analysis to real-world examples like the
\emph{JavaScript Cookie Library with jQuery bindings and JSON
  support}\footnote{\webpageCOOKIE}  (version 2.2.0) and the
\emph{Rye}\footnote{\webpageRYE} library (version 0.1.0), a JavaScript
library for DOM manipulation. 

These libraries were augmented by wrapping several functions and
objects using the \lstinline{trace} function. The analysis
successfully tracks the flow of the thus marked values, which pop up
in the expected places.




\section{Dependency Classification} \label{sec:untrace}


\begin{figure}
                \centering
                \begin{displaymath}
                                \begin{array}{lrl}
                                                \ljExp &::=& \dots ~|~ \ljTraceClass(\ljExp,~ \ljConst) ~|~ \ljUntraceClass(\ljExp,~ \ljConst)
                                \end{array}
                \end{displaymath}
                \caption{Extended syntax of $\ldcj$.}
                \label{fig:syntax_ldcj}
\end{figure}

\begin{figure}
                \centering
                \begin{mathpar}
                                \inferrule [\LDJTraceExt]
                                {
                                                \ljHeap,\ljEnv,\ljType\joinType\ljSource^{\ljConst} ~\entails~ \ljExp ~\eval~ \ljHeap' ~|~ \ljVal:\ljType_{\ljVal}
                                }
                                {
                                                \ljHeap,\ljEnv,\ljType ~\entails~ \ljTraceClass~ (\ljExp,~ \ljConst) ~\eval~ \ljHeap' ~|~ \ljVal:\ljType_{\ljVal}
                                }\and
                                \inferrule [\LDJUntrace]
                                {
                                                \ljHeap,\ljEnv,\ljType ~\entails~ \ljExp ~\eval~ \ljHeap' ~|~ \ljVal:\ljType_{\ljVal}\\
                                                \ljType' = \ljType_{\ljVal}[\ljSLocation^{\ljClass,\ljConst}\mapsto\ljSLocation^{\ljClass',\ljConst}]
                                }
                                {
                                                \ljHeap,\ljEnv,\ljType ~\entails~ \ljUntraceClass~ (\ljExp,~  \ljConst) ~\eval~ \ljHeap' ~|~ \ljVal:\ljType'
                                }
                \end{mathpar}
                \caption{Inference rules of $\ldcj$.}
                \label{fig:inference-rules_ldcj}
\end{figure}

\begin{figure}
                \centering
                \begin{mathpar}
                                \inferrule [\ATraceExt]
                                {
                                                \aState[\D\mapsto\D_{\aState}\join\ljSource^{\ljClass,\ljConst}], \aScope ~\entails~ \ljExp ~\eval~ \aState' ~|~ \aVal
                                }
                                {
                                                \aState,\aScope ~\entails~ \ljTraceClass~ (\ljExp,~ \ljConst) ~\eval~ \langle \aObjStore_{\aState'}, \D_{\aState} \rangle ~|~ \aVal
                                }\and
                                \inferrule [\AUntrace]
                                {
                                                \aState, \aScope ~\entails~ \ljExp ~\eval~ \aState' ~|~ \aVal\\
                                                \D_{\aVal}'=\D_{\aVal}[\ljSLocation^{\ljClass,\ljConst}\mapsto\ljSLocation^{\ljClass',\ljConst}]
                                }
                                {
                                                \aState,\aScope ~\entails~ \ljUntraceClass~ (\ljExp,~ \ljConst) ~\eval~ \langle \aObjStore_{\aState'}, \D_{\aState} \rangle ~|~
                                                \langle \aLattice_{\aVal}, \aObjlabel_{\aVal}, \D_{\aVal}' \rangle
                                }
                \end{mathpar}
                \caption{Inference rules for abstract trace.}
                \label{fig:inference-rules_trace}
\end{figure}

To cater for dependency classification,
the accompanying formal framework $\ldcj$ extends $\ldj$ (Figure
\ref{fig:syntax_ldcj}). In  $\ldcj$ marks
are classified according to a 
finite set of modes. They are further augmented by an identifier
that can be referred to in the \lstinline'trace' and
\lstinline'untrace' expressions. The operator $\ljTraceClass$ 
generates a mark in mode $\ljClass$ and the \lstinline'untrace' operator
changes the mode of all $\ell$-marks according to the sanitization
method applied (this distinction is ignored in the example). In the
calculus, this change is expressed by the $\ljUntraceClass$
expression, where $\ljClass$ ranges over an unspecified set of modes.

Marks $\ljType ::= \dots ~|~ \ljSLocation^{\ljClass,\ljConst}$ are
extended by an new mark-type, a location classified with a class
$\ljClass$ and identifier $\ljConst$. 

The mark propagation is like in Section~\ref{sec:dependency_type} (see
Figure \ref{fig:inference-rules_ldcj}). Rule \Rule{\LDJTraceExt}
augments the sub-context with the new classified mark.
\Rule{\LDJUntrace} substitutes location
$\ljSLocation^{\ljClass,\ljConst}$ by a declassified location
$\ljSLocation^{\ljClass',\ljConst}$.

In the analysis, $\dTrace ::= \ljSLocation ~|~
\ljSLocation^{\ljClass,\ljConst}$ replaces $\ljSLocation$ in $\D$. 
Rule \Rule{\ATraceExt} (Figure \ref{fig:inference-rules_trace})
generates new dependencies and \Rule{\AUntrace} substitutes $\ljClass$
by $\ljClass'$ in all locations $\ljSLocation$ labeled with $\ljConst$.


\section{Technical Results}
\label{sec:technical_results}

To prove the soundness of our abstract analysis we show termination insensitive noninterference.  The required steps are
proving noninterference for the $\ldj$ calculus, showing that the abstract analysis provides a correct abstraction of
the $\ldj$ calculus, and that the abstract analysis
terminates.

%
%

\subsection{Noninterference}
\label{sec:noninterference}

Proving noninterference requires relating different substitution
instances of the same expression.  As they may evaluate differently,
we need to be able to cater for differences in the heap, for example,
with respect to locations.

\begin{definition}\label{def:bijection}
  A renaming $\bijection ::= \emptyset \mid \bijection[\ljLocation\mapsto\ljLocationPrime]$
  is a partial mapping on locations where
  $\bijection (\ljLocation)$ carries the same mark
  $\ell$ as $\ljLocation$.

  It extends to values by $\bijection (\ljConst) = \ljConst$.
\end{definition}

In the upcoming definitions, the dependency annotation $\ljType$
contains the marks created by the selected $\ljTrace$ 
expression, the body of which may be substituted.

Further, we introduce equivalence relations for each element affected by the $\ljSource$ substitution.

\begin{definition}\label{def:k-equivalence_value}
                Two marked values are $\bijection,\ljType$-equivalent
                $\ljVal_{0}:\ljType_{0}\equivType\ljVal_{1}:\ljType_{1}$ if they are equal as long as their marks are disjoint
                from $\ljType$. 
                \begin{align}
                                \ljType\cap\ljType_{0}=\emptyset ~\wedge~ \ljType\cap\ljType_{1}=\emptyset
                                ~\Rightarrow~ \bijection(\ljVal_{0})=\ljVal_{1}
                \end{align}
\end{definition}

\begin{definition}\label{def:k-equivalence_environment}
                Two environments $\ljEnv_{0}$,$\ljEnv_{1}$ are $\bijection,\ljType$-equivalent $\ljEnv_{0}\equivType\ljEnv_{1}$
                if $R := \dom (\ljEnv_0) = \dom (\ljEnv_1)$ and they contain equivalent values.
                \begin{align}
                                &\forall \ljVar \in R:~ \ljEnv_{0}(\ljVar)
                                \equivType \ljEnv_{1}(\ljVar)
                \end{align}
\end{definition}

\begin{definition}\label{def:k-equivalence_expression}
                Two expressions $\ljExp_{0}$,$\ljExp_{1}$ are $\bijection,\ljType$-equivalent $\ljExp_{0}\equivType\ljExp_{1}$
                iff they only differ in the argument of $\ljTrace(\ljExp')$ subexpressions with $\ljSource\in\ljType$.
                \begin{align}
                                \begin{split}
                                                &\ljType = \{\ljSource_{0}, ..., \ljSource_{n}\} ~\Rightarrow~\\
                                                &\exists \ljExp_{0}'\ldots\exists \ljExp_{n}':~ \ljExp_{0}
                                                ~=~
                                                \ljExp_{1}[\ljSource_{0}\mapsto\ljExp_{0}']\ldots[\ljSource_{n}\mapsto\ljExp_{n}']
                                \end{split}
                \end{align}
\end{definition}

\begin{definition}\label{def:k-equivalence_object}
                Two storables $\ljStorable_{0}$,$\ljStorable_{1}$ are $\bijection,\ljType$-equivalent\\ $\langle \ljObj_{0},
                \langle\ljEnv_{0},\ljFunc\ljVar.\ljExp_{0}\rangle, \ljProto_{0} \rangle \equivType \langle \ljObj_{1},
                \langle\ljEnv_{1},\ljFunc\ljVar.\ljExp_{1}\rangle, \ljProto_{1} \rangle$ if $S:= \dom (\ljObj_0) = \dom (\ljObj_1)$ and they
                only differ in values $\ljConst:\ljType_{\ljConst}$ with any intersection with $\ljType$.
                \begin{align}
                                \begin{split}
                                                &\forall \ljStr \in S:~ \ljObj_{0}(\ljStr) \equivType \ljObj_{1}(\ljStr)
                                \end{split}
                                \\
                                \begin{split}
                                                &\ljEnv_{0}\equivType\ljEnv_{1}
                                                ~\wedge~
                                                \ljFunc\ljVar.\ljExp_{0}\equivType\ljFunc\ljVar.\ljExp_{1}
                                \end{split}\\
                                \begin{split}
                                                &\bijection(\ljProto_{0})=\ljProto_{1}
                                \end{split}
                \end{align}
\end{definition}

\begin{definition}\label{def:k-equivalence_heap}
                Two heaps $\ljHeap_{0}$,$\ljHeap_{1}$ are $\bijection,\ljType$-equivalent $\ljHeap_{0}\equivType\ljHeap_{1}$ if
                they only differ in values $\ljVar:\ljType_{\ljVar}$ with any intersection with $\ljType$ or in one-sided
                locations.
                \begin{align}
                                \forall \ljLocation \in \dom(\bijection):~ \ljHeap_{0}(\ljLocation) \equivType
                                \ljHeap_{1}(\bijection(\ljLocation))
                \end{align}
\end{definition}

Now, the noninterference theorem can be stated as follows. 

\begin{theorem}\label{thm:noninterference}
                Suppose $\ljHeap,\ljEnv,\ljType \entails \ljExp \eval \ljHeap' \mid v:\ljType_{v}$.
                If $\ljSource \notin \bar{\ljType}$ and $\ljHeap \equivTypeSub \tilde{\ljHeap}$ and $\ljEnv\equivTypeSub\tilde{\ljEnv}$ then
                $\tilde{\ljHeap},\tilde{\ljEnv},\ljType \entails \bar{\ljExp} \eval \tilde{\ljHeap}' \mid
                \tilde{v}:\tilde{\ljType}_{v}$ with $\bar{\ljExp}=\ljExp[\ljSource \mapsto \tilde{\ljExp}]$ and
                $\ljExp \equivTypeSub \bar{\ljExp}$ and $\ljHeap' \equivTypeSubPrime \tilde{\ljHeap}'$ and
                $v:\ljType_{v}\equivTypeSubPrime\tilde{v}:\tilde{\ljType}_{v}$,
                for some $\bijection'$ extending $\bijection$.
\end{theorem}

The proof is by induction on the evaluation $\eval$ (Section \ref{sec:proof_noninterference}).

%
%

\subsection{Correctness}
\label{sec:correctness}

The abstract analysis is a correct abstraction of the $\ldj$ calculus.
To formalize correctness, we introduce a consistency
relation that relates semantic domains of the concrete dependency
tracking semantics of $\ldj$ with the abstract domains.

\begin{definition}\label{def:consistency_relation} The consistency relation $\equivLattice$ is defined by:
                \begin{displaymath}
                                \begin{array}{lrl}

                                                \ljType\equivLattice\D &\Leftrightarrow& \ljType\subseteq\D\\

                                                \ljConst\equivLattice\aLattice &\Leftrightarrow& \ljConst \in \aLattice\\

                                                \ljLocation\equivLattice\aObjlabel &\Leftrightarrow& \ljSLocation \in \aObjlabel\\

                                                \ljVal\equivLattice\aVal &\Leftrightarrow&
                                                \begin{cases}
                                                                \ljLocation\equivLattice\aObjlabel_\aVal, & \ljVal=\ljLocation\\
                                                                \ljConst\equivLattice\aLattice_\aVal, & \ljVal=\ljConst
                                                \end{cases}\\

                                                \ljVal:\ljType\equivLattice\aVal &\Leftrightarrow& \ljType \equivLattice \D_{\aVal} ~\wedge~
                                                \ljVal\equivLattice\aVal\\

                                                \ljObj\equivLattice\aObj  &\Leftrightarrow& \forall \ljStr \in \dom(\ljObj):~ \exists \aLattice
                                                \in \dom(\aObj):~ \\
                                                &&\ljStr\equivLattice\aLattice ~\wedge~
                                                \ljObj(\ljStr)\equivLattice\aObj(\aLattice) ~\wedge~\\
                                                &&\forall \ljStr \notin dom(\ljObj):~ \ljUndefined \equivLattice\aObj(\aLattice)\\

                                                \ljEnv \equivLattice \aScope &\Leftrightarrow& \forall \ljVar \in \dom(\ljEnv):~ \ljVar \in
                                                \dom(\aScope) \\
                                                                                                &&\wedge~ \ljEnv(\ljVar)\equivLattice\aScope(\ljVar)\\

                                                \ljClosure\equivLattice\aFunc &\Leftrightarrow& \langle\ljEnv,\ljFunc\ljVar.\ljExp_{\ljClosure}\rangle=\ljClosure ~\wedge~ 
                                                \langle\aScope,\ljFunc\ljVar.\ljExp_{\aFunc}\rangle=\aFunc\\ 
                                                &&\rightarrow~ \ljEnv \equivLattice \aScope ~\wedge~ \ljFunc\ljVar.\ljExp_{\ljClosure}=\ljFunc\ljVar.\ljExp_{\aFunc}\\

                                                \ljStorable \equivLattice \aStorable &\Leftrightarrow&
                                                \langle\ljObj,\ljClosure,\ljProto\rangle=\ljStorable ~\wedge~ \langle\aObj,\aFunc,\aObjlabel\rangle=\aStorable\\
                                                &&~\rightarrow~ \ljObj\equivLattice\aObj ~\wedge~ \ljClosure\equivLattice\aFunc ~\wedge~ \ljProto \equivLattice \aObjlabel\\

                                                \ljHeap\equivLattice\aState &\Leftrightarrow& \forall \ljLocation \in \dom(\ljHeap):~
                                                \ljSLocation \in \aObjStore_\aState\\
                                                                                                &&\wedge~ \ljHeap(\ljLocation)\equivLattice\aObjStore_\aState(\ljSLocation)

                                \end{array}
                \end{displaymath}
\end{definition}

Showing adherence to the inference of $\ldj$ requires to proof that consistent heaps, environments, and values
produce a consistent heap and value.

\begin{lemma}[Program]\label{thm:program} $\forall \ljExp:$
                $\emptyset, \emptyset, \emptyset \entails \ljExp ~\eval~ \ljHeap ~|~ \ljTypedVal$ and $\entails \ljExp ~\eval~ \aState ~|~ \aVal$ implies that 
                $\langle \ljHeap, \ljTypedVal \rangle \equivLattice \langle \aState, \aVal \rangle$
\end{lemma}

  Given by theorem \eqref{thm:correctness} and definition \eqref{def:consistency_relation}.

\begin{lemma}[Property Reference]\label{thm:property_reference}
                $\forall \ljHeap, \ljLocation, \ljStr, \aState, \aObjlabel, \aLattice:$ 
                $\ljHeap \equivLattice \aState$, $\ljLocation \equivLattice \aObjlabel$, $\ljStr \equivLattice \aLattice$, and
                $\entailsPRIteration \aState(\aObjlabel), \aLattice ~\eval~ \aVal$ implies
                $\ljHeap(\ljLocation)(\ljStr) ~\equivLattice~ \aVal$
\end{lemma}

                                The proof is by definition \eqref{def:consistency_relation} and by induction on the
                derivation of $\aState'' ~\entailsPRIteration~ \aState''(\aObjlabel_{0}),\lvString ~\eval~ \aVal$.

\begin{lemma}[Property Assignment]\label{thm:property_assignment}
                $\forall \ljHeap, \ljLocation, \ljStr, \ljTypedVal, \aState, \aObjlabel, \aLattice, \aVal:$
                $\ljHeap \equivLattice \aState$, $\ljLocation \equivLattice \aObjlabel$, $\ljStr \equivLattice \aLattice$, $\ljTypedVal \equivLattice \aVal$
                and $\aState \entailsPAIteration \aObjlabel, \aLattice, \aVal ~\eval~ \aState'$ implies
                $\ljHeap[\ljLocation,\ljStr \mapsto \ljTypedVal] ~\equivLattice~ \aState'$
\end{lemma}

                                The proof is by definition \eqref{def:consistency_relation} and by induction on the
                derivation of $\aState''' ~\entailsPAIteration~ \aObjlabel_{0},\lvString,\langle \aLattice_{\aVal}, \aObjlabel_{\aVal}, \D_{0} \join \D_{1} \join \D_{\aVal} \rangle ~\eval~ \aState''''$.

The following correctness theorem relates the concrete semantics to the abstract semantics. 

\begin{theorem} \label{thm:correctness}
                Suppose that $\ljHeap,\ljEnv,\ljType ~\entails~ \ljExp ~\eval~ \ljHeap' ~|~ \ljVar$ then
                $\forall \aState,\aScope$ with $\ljHeap\equivLattice\aState$, $\ljEnv\equivLattice\aScope$ and
                $\ljType\equivLattice\D_{\aState}$: $\aState,\aScope ~\entails~ \ljExp ~\eval~ \aState' ~|~ \aVal$ with
                $\ljHeap'\equivLattice\aState'$ and $\ljVar\equivLattice\aVal$. 
\end{theorem}

The proof is by induction on the evaluation of $\ljExp$ (Section \ref{sec:proof_correctness}).

%
%

\subsection{Termination}
\label{sec:termination}

Finally, we want to guarantee termination of our analysis.

\begin{theorem} \label{thm:termination}
                For each $\aState$, $\aScope$, and $\ljExp$, there exist $\aState'$ and $\aVal$
                such that $\aState,\aScope ~\entails~ \ljExp ~\eval~ \aState' ~|~ \aVal$. 
\end{theorem}

For the proof (Section \ref{sec:proof_termination}), we observe that all rules of the abstract system in Section~\ref{sec:abstract_analysis} are monotone with
respect to all their inputs.  As the analysis lattice for $\aState$ has finite height, all fixpoint computations in the
abstract semantics terminate.


\section{Implementation}
\label{sec:implementation}

The implementation extends TAJS\footnote{\webpageTAJS}, the type
analyzer for JavaScript.  TAJS accepts standard JavaScript \cite{ecma1999:262} .  

The abstract interpretation of values and the analysis state are extended by a set of dependency annotations, according
to the description in Section~\ref{sec:abstract_analysis}.  As shown in Section~\ref{sec:syntax_ldj} values can be
marked by using the $\ljTrace$ expression, which is implemented as a built-in function.  A configuration file can be used
to trace values produced by JavaScript standard operations or DOM
functions. The DOM environment gets constructed during the
initialization of TAJS and is available as as normal code would
be. The functionality and the dependency propagation for these
operations is hard coded.

The extended dependency set has no influence on the lattice structure and does not compromise the precision of the type analyzer. Some notes about the precision can be found in the original work of TAJS \cite{tajs2009}.

The functions $\ljTrace$ and $\ljUntraceClass$ have to be
defined as identity functions before the instrumented code can run in
a standard JavaScript engine.



TAJS handles all language features like prototypes, iterations, and exceptions.
The specification in Figure~\ref{fig:inference-rules_abstract}
simplifies the implementation in several respects.
To support the conditional to properly account for indirect
information flow, the control flow graph had to be extended with special
dependency push and pop nodes to encapsulate sub-graphs and to add or
remove state dependencies.

The type analyzer provides an over-approximation according to the principles described in
Section~\ref{sec:formalization}.  The analysis result shows the set of
traced values and the set of values, 
which are potentially influenced by them.

There are several ways to use the analyzer.  First, a value can be
marked and  its influence and usage can be determined.  This feature can be used to prevent private data from illegal usage and theft.  Second, the $\ljTrace$
function may be used to encapsulate foreign code.  As a result of this encapsulation each value which is modified due to
the foreign code is highlighted by the analysis.  An inspection of the
results can show breaches of integrity. 


Our implementation is based on an early version of TAJS. The current
TAJS version includes support  for further language features including \texttt{eval}
\cite{unevalizer2012}. The dependency analysis can benefit from these
extensions by merging it into the current development branch of TAJS.

%

\subsection{Runtime Evaluation}

We evaluated the performance impact of our extension by analyzing
programs from the Google V8 Benchmark
Suite\footnote{\webpageGOOGLE}. The programs we selected range from
about 400 to 5000 lines of code and perform tasks like an OS kernel 
simulation, constraint solving, or extraction of regular
expressions. The tests were run on a MacBook Pro with 2 GHz Intel
Core i7 processor with 8 GB memory.

\begin{figure}
\centering
\small
\begin{tabular}{ l || l | l}
\toprule
\textbf{Benchmark} & \textbf{TAJS} & \textbf{TbDA} \\ \midrule
Richards \scriptsize(539 lines)         & 1596                                  & 2890 \\
DeltaBlue \scriptsize(880 lines)                & 3471                                  & 4031 \\
Crypto \scriptsize(1689 lines)                  & 3637                                  & 7527 \\
RegExp \scriptsize(4758 lines)                  & 3710                                  & 4104 \\
Splay \scriptsize(394 lines)                    & 1598                                  & 2521 \\
Navier Stokes \scriptsize(387 lines)    & 2794                                  & 3118 \\
\bottomrule
\end{tabular}
\caption{Google V8 Benchmark Suite.}
\label{fig:googleV8Runtime}
\end{figure}

Figure~\ref{fig:googleV8Runtime} shows the particular benchmarks
together with the averaged time (in milliseconds) to run the analysis
and to print the output. The \textbf{TAJS} column shows the 
timing of the original type analyzer without dependency extension. The
\textbf{TbDA} column shows the timing of our extended
version. The figures demonstrate that the dependency analysis leads to
a slowdown between 12\% and 106\%. Two further benchmarks
(\emph{RayTrace} and \emph{EarleyBoyer}) did not run to completion 
because of compatibility problems caused by the outdated version of
TAJS underlying our implementation.


\section{Related Work} \label{sec:related_work}

Information flow analysis was pioneered by Denning's work
\cite{Denning:1976:LMS:360051.360056,Denning:1977:CPS:359636.359712}
which models different security levels as values in a lattice
containing elements like \emph{High} and \emph{Low} and
which suggests an analysis as an abstract 
interpretation of the propagation of these levels through the program.  Zanioli and others
\cite{Zanioli:2012:SSA:2231936.2231983} present a recent example of such an analysis with an emphasis on constructing an
expressive analysis domain. 

Many authors have taken up this approach and transposed it to type
theoretic and logical settings
\cite{Heintze:1998:SCP:268946.268976,Abadi:2007:ACC:1230168.1230577,Volpano:1996:STS:353629.353648,amtoft2004information,Sabelfeld03language-basedinformation-flow}.

In these systems, input value types are enhanced with security levels.  Well-typed
programs guarantee that no \emph{High} value flows into a \emph{Low} output and thus
noninterference between high inputs and low outputs \cite{GoguenMeseguer1982}.  Similar to the soundness property of our dependency analysis changes on
\emph{High} inputs are unobservable in \emph{Low} outputs.
Dependencies are related to security types, but more
flexible 
\cite{AbadiBanerjeeHeintzeRiecke1999}.  They can be analyzed before committing to a fixed security lattice.

Security aspects of JavaScript programs have received much attention. Different approaches focus on static or dynamic analysis
techniques, e.g. \cite{Just:2011:IFA:2093328.2093331,
  hedin2012information, Chugh:2009:SIF:1542476.1542483}, or attempt to
make guarantees by
reducing the functionality \cite{miller2008safe}. The analysis for dependencies is no security analysis per se,
but the analysis results express information that is relevant for
confidentiality and integrity concerns.

Dependency analysis can be seen as the static counterpart to data
tainting (e.g.,  \cite{Clause:2007:DGD:1273463.1273490}).
Tainting relies on augmenting the run-time representation of a value with
information about its properties (like its confidentiality level).  Users of the value first check at run time if that use is granted according to some
security policy.  Dynamic tainting approaches have been successfully used to address security attacks, including buffer
overruns, format string attacks, SQL and command injections, and cross-site scripting. Tainting semantics are also used
for automatic sanitizer placement \cite{Livshits:2013:TFA:2429069.2429115,DBLP:conf/sp/BalzarottiCFJKKV08}. There are also uses in program understanding,
software testing, and debugging. Tainting can also be augmented with static analysis to increase its effectiveness \cite{vogt2007cross}.

Dynamic languages like JavaScript have many peculiarities that make program
analysis and the interpretation of its results challenging \cite{5230486}. 
TAJS \cite{tajs2009} and hence our analysis can handle almost all dynamic features of JavaScript.


\section{Conclusion}
\label{sec:conclusion}

We have designed a type-based dependency analysis for JavaScript, proved its soundness and termination, and
demonstrated that independence ensures noninterference.  We have implemented the analysis as an
extension of the open-source JavaScript analyzer TAJS.  This approach ensures that our analysis can be applied
to real-world JavaScript programs.

While a dependency analysis is not a security analysis,  it can form the basis for investigating
noninterference.  This way, its results can be used to ensure confidentiality and
integrity, as well as verify the correct placement of sanitizers.




\bibliographystyle{abbrvnat}

\balancecolumns


\appendix
%

%

%



\section{Context Dependency}
\label{sec:proof_context}

\begin{proof}[Proof of Lemma \ref{thm:context_dependency}]


 $\ljHeap,\ljEnv,\ljType \entails \ljExp \eval \ljHeap' \mid \ljVal_{\ljExp}:\ljType_{\ljExp}$
   implies that $\ljType\subseteq\ljType_{\ljExp}$.
   By induction on the derivation of
   $\ljHeap,\ljEnv,\ljType ~\entails~ \ljExp ~\eval~ \ljHeap' ~|~ \ljVal_{\ljExp}:\ljType_{\ljExp}$.

   \begin{description}
	  \item[Case] \Rule{\LDJConstant}: $\ljExp\equiv\ljConst;~ \ljVal_{\ljExp}\equiv\ljConst;~ \ljType_{\ljExp}\equiv\ljType$. 
		 Claim holds because $\ljType\subseteq\ljType$.

	  \item[Case] \Rule{\LDJVariable}: $\ljExp\equiv\ljVar;~ \ljVal_{\ljExp}\equiv\ljVal;~ \ljType_{\ljExp}\equiv\ljType_{\ljVal}\joinType\ljType ~|~ \ljVal:\ljType_{\ljVal} = \ljEnv(\ljVar)$. 
		 Claim holds because $\ljType\subseteq\ljType_{\ljVal}\joinType\ljType$.

	  \item[Case] \Rule{\LDJOperation}:
		 $\ljExp\equiv\ljOp(\ljExp_{0}, \ljExp_{1});~ \ljVal_{\ljExp}\equiv\ljVal_{op};~ \ljType_{\ljExp}\equiv\ljType_0\joinType\ljType_1$.
		 By inversion 
		 \begin{align}
			\ljHeap,\ljEnv,\ljType ~&\entails~  \ljExp_0 ~\eval~ \ljHeap' ~|~ \ljVal_0:\ljType_{0}\\
			\ljHeap',\ljEnv,\ljType ~&\entails~  \ljExp_1 ~\eval~ \ljHeap'' ~|~ \ljVal_1:\ljType_{1}
		 \end{align}
		 holds. Induction yields $\ljType\subseteq\ljType_{0}$ and $\ljType\subseteq\ljType_{1}$.
		 Claim holds because $\ljType\subseteq\ljType_{\ljExp}=\ljType_0\joinType\ljType_1$.

	  \item[Case] \Rule{\LDJFunctionCreation}: $\ljExp\equiv\ljFunc \ljVar.\ljExp;~ \ljVal_{\ljExp}\equiv\ljLocation;~ \ljType_{\ljExp}\equiv\ljType$.
		 Claim holds because $\ljType\subseteq\ljType$.

	  \item[Case] \Rule{\LDJObjectCreation}: $\ljExp\equiv\ljNew~\ljExp_{p};~ \ljVal_{\ljExp}\equiv\ljLocation;~ \ljType_{\ljExp}\equiv\ljType_{\ljVal}$.
		 By inversion 
		 \begin{align}
			\ljHeap,\ljEnv,\ljType ~\entails~ \ljExp_{p} ~\eval~ \ljHeap' ~|~ \ljVal:\ljType_{\ljVal}
		 \end{align}
		 holds. Induction yields $\ljType\subseteq\ljType_{\ljVal}$.
		 Claim holds because $\ljType\subseteq\ljType_{\ljVal}$.

	  \item[Case] \Rule{\LDJFunctionApplication}: $\ljExp\equiv\ljExp_0(\ljExp_1);~ \ljVal_{\ljExp}\equiv\ljVal;~ \ljType_{\ljExp}\equiv\ljType_{\ljVal}$.
		 By inversion 
		 \begin{align}
			\ljHeap'',\dot{\ljEnv}[\ljVar \mapsto \ljVal_{1}:\ljType_{1}],\ljType\joinType\ljType_{0} ~\entails~ \dot{\ljExp} ~\eval~ \ljHeap''' ~|~ \ljVal:\ljType_{\ljVal}
		 \end{align}
		 holds where
		 $\dot{\ljEnv}, \dot{\ljExp}$ is the
		 closure obtained
		 from evaluating $\ljExp_0$. It follows that $\ljType\subseteq\ljType\joinType\ljType_{0}\subseteq\ljType_{\ljVal}$.

	  \item[Case] \Rule{\LDJPropertyReference}: $\ljExp\equiv\ljExp_0[\ljExp_1];~ \ljVal_{\ljExp}\equiv\ljVal;~ \ljType_{\ljExp}\equiv\ljType_{\ljVal}\joinType\ljType_{\ljLocation}\joinType\ljType_{\ljStr} ~|~ \ljVal:\ljType_{\ljVal} = \ljHeap''(\ljLocation)(\ljStr)$.
		 By inversion 
		 \begin{align}
			\ljHeap,\ljEnv,\ljType ~\entails~ \ljExp_{0} ~\eval~ \ljHeap' ~|~ \ljLocation:\ljType_{\ljLocation} \\
			\ljHeap',\ljEnv,\ljType ~\entails~ \ljExp_{1} ~\eval~ \ljHeap'' ~|~ \ljStr:\ljType_{\ljStr}
		 \end{align}
		 holds.  Induction yields $\ljType\subseteq\ljType_{\ljLocation}$ and $\ljType\subseteq\ljType_{\ljStr}$.
		 Claim holds because $\ljType\subseteq\ljType_{\ljVal}\joinType\ljType_{\ljLocation}\joinType\ljType_{\ljStr}$.

	  \item[Case] \Rule{\LDJPropertyAssignment}: $\ljExp\equiv\ljExp_0[\ljExp_1] = \ljExp_2;~ \ljVal_{\ljExp}\equiv\ljVal;~ \ljType_{\ljExp}\equiv\ljType_{\ljVal}$.
		 By inversion 
		 \begin{align}
			\ljHeap'',\ljEnv,\ljType ~\entails~ \ljExp_{2} ~\eval~ \ljHeap''' ~|~ \ljVal:\ljType_{\ljVal}
		 \end{align}
		 holds. Induction yields $\ljType\subseteq\ljType_{\ljVal}$.
		 Claim holds because $\ljType\subseteq\ljType_{\ljVal}$.

	  \item[Case] \Rule{\LDJConditionTrue}: $\ljExp\equiv\ljIf~ (\ljExp_0)~ \ljExp_1,~ \ljExp_2;~ \ljVal_{\ljExp}\equiv\ljVal_1;~ \ljType_{\ljExp}\equiv\ljType_{1}$.
		 By inversion 
		 \begin{align}
			\ljHeap',\ljEnv,\ljType\joinType\ljType_{0} ~\entails~ \ljExp_1 ~\eval~ \ljHeap_{1}'' ~|~ \ljVal_{1}:\ljType_{1}
		 \end{align}
		 holds.
		 It follows that $\ljType\subseteq\ljType\joinType\ljType_{0}\subseteq\ljType_{1}$

	  \item[Case] \Rule{\LDJConditionFalse}: Analogous to case\\
		 \Rule{\LDJConditionTrue}.

	  \item[Case] \Rule{\LDJTrace}: $\ljExp\equiv\ljTrace~ (\ljExp');~ \ljVal_{\ljExp}\equiv\ljVal;~ \ljType_{\ljExp}\equiv\ljType_{\ljVal}$.
		 By inversion 
		 \begin{align}
			\ljHeap,\ljEnv,\ljType\joinType\ljSource ~\entails~ \ljExp' ~\eval~ \ljHeap' ~|~ \ljVal:\ljType_{\ljVal}
		 \end{align}
		 holds. It follows that $\ljType\subseteq\ljType\joinType\ljSource\subseteq\ljType_{\ljVal}$

   \end{description}
\end{proof}


\section{Noninterference}
\label{sec:proof_noninterference}

%

\begin{proof}[Proof of Theorem \ref{thm:noninterference}]

   Suppose $\ljHeap,\ljEnv,\ljType \entails \ljExp \eval \ljHeap' \mid \ljVal:\ljType_{\ljVal}$.
   If $\ljSource \notin \bar{\ljType}$ and $\ljHeap \equivTypeSub \tilde{\ljHeap}$ and $\ljEnv\equivTypeSub\tilde{\ljEnv}$ then
   $\tilde{\ljHeap},\tilde{\ljEnv},\ljType \entails \bar{\ljExp} \eval \tilde{\ljHeap}' \mid
   \tilde{\ljVal}:\tilde{\ljType}_{\ljVal}$ with $\bar{\ljExp}=\ljExp[\ljSource \mapsto \tilde{\ljExp}]$ and
   $\ljExp \equivTypeSub \bar{\ljExp}$ and $\ljHeap' \equivTypeSubPrime \tilde{\ljHeap}'$ and
   $v:\ljType_{\ljVal}\equivTypeSubPrime\tilde{\ljVal}:\tilde{\ljType}_{\ljVal}$,
   for some $\bijection'$ extending $\bijection$.

   Proof by induction on the derivation of
   $\ljHeap,\ljEnv,\ljType ~\entails~ \ljExp ~\eval~ \ljHeap' ~|~ \ljVal:\ljType$ and 
   $\tilde{\ljHeap},\tilde{\ljEnv},\ljType ~\entails~ \bar{\ljExp} ~\eval~ \tilde{\ljHeap}' ~|~ \tilde{\ljVal}:\tilde{\ljType}_{\ljVal}$.

   \begin{description}


	  \item[Case] \Rule{\LDJConstant}: $\ljExp\equiv\ljConst$.
		 By definition \eqref{def:k-equivalence_value} claim holds because:
		 $\ljConst = \ljConst ~\rightarrow~ \ljConst:\ljType\equivTypeSub\ljConst:\ljType$


	  \item[Case] \Rule{\LDJVariable}: $\ljExp\equiv\ljVar$.
		 $\forall\ljSource\notin\bar{\ljType}$ yields
		 \begin{description}
			\item[Subcase] $\ljSource \notin \ljType$:
			   By definition \eqref{def:k-equivalence_value}, \eqref{def:k-equivalence_environment} claim holds because:
			   $ 
			   \ljEnv \equivTypeSub \tilde{\ljEnv} ~\rightarrow~ 
			   \ljEnv(\ljVar)\joinType\ljType\equivTypeSub\tilde{\ljEnv}(\ljVar)\joinType\ljType
			   $

			\item[Subcase] $\ljSource \in \ljType$:
			   From definition \eqref{def:k-equivalence_value}, \eqref{def:k-equivalence_environment} claim holds because:
			   $
			   \ljSource \in \ljType ~\rightarrow~ \ljEnv(\ljVar)\joinType\ljType\equivTypeSub\tilde{\ljEnv}(\ljVar)\joinType\ljType
			   $
		 \end{description}


	  \item[Case] \Rule{\LDJOperation}: $\ljExp\equiv\ljOp(\ljExp_{0}, \ljExp_{1})$.
		 By inversion 
		 \begin{align}
			&\forall \ljSource \notin \bar{\ljType}:~ \tilde{\ljHeap},\tilde{\ljEnv},\ljType ~\entails~
			\ljExp_{0}[\ljSource\mapsto\tilde{\ljExp}] ~\eval~
			\tilde{\ljHeap}' ~|~ \tilde{\ljVal}_{0}:\tilde{\ljType}_{0}\\
			&\forall \ljSource \notin \bar{\ljType}:~ \tilde{\ljHeap}',\tilde{\ljEnv},\ljType ~\entails~
			\ljExp_{1}[\ljSource\mapsto\tilde{\ljExp}] ~\eval~
			\tilde{\ljHeap}'' ~|~ \tilde{\ljVal}_{1}:\tilde{\ljType}_{1}
		 \end{align}
		 holds. $\forall\ljSource\notin\bar{\ljType}$ yields
		 \begin{description}
			\item[Subcase] $\ljSource \notin \ljType_{0} ~\wedge~
			   \ljSource \notin \tilde{\ljType}_{0} ~\wedge~
			   \ljSource \notin \ljType_{1} ~\wedge~
			   \ljSource \notin \tilde{\ljType}_{1}$:
			   By definition \eqref{def:k-equivalence_value} we obtain
			   \begin{align}
				  \ljSource \notin \ljType_{0} ~\wedge~ \ljSource \notin \tilde{\ljType}_{0} ~&\rightarrow~ \ljVal_{0}=\tilde{\ljVal}_{0}\\
				  \ljSource \notin \ljType_{1} ~\wedge~ \ljSource \notin \tilde{\ljType}_{1} ~&\rightarrow~ \ljVal_{1}=\tilde{\ljVal}_{1}
			   \end{align}
			   and this leads to:
			   \begin{align}
				  \begin{split}
					 &\ljOperation(\ljVal_0, \ljVal_1)\joinType\ljType_{0}\joinType\ljType_{1}
					 \equivTypeSub\\
					 &\ljOperation(\tilde{\ljVal}_0, \tilde{\ljVal}_1)\joinType\tilde{\ljType}_{0}\joinType\tilde{\ljType}_{1}
				  \end{split}
			   \end{align}

			\item[Subcase] $\ljSource \in \ljType_{0} ~\vee~
			   \ljSource \in \tilde{\ljType}_{0} ~\vee~ 
			   \ljSource \in \ljType_{1} ~\vee~
			   \ljSource \in \tilde{\ljType}_{1}$:
			   Again by~\eqref{def:k-equivalence_value}
			   \begin{align}
				  \begin{split}
					 &\ljOperation(\ljVal_0, \ljVal_1)\joinType\ljType_{0}\joinType\ljType_{1}
					 \equivTypeSub\\
					 &\ljOperation(\tilde{\ljVal}_0, \tilde{\ljVal}_1)\joinType\tilde{\ljType}_{0}\joinType\tilde{\ljType}_{1}
				  \end{split}
			   \end{align}     
			   holds.
		 \end{description}
		 Claim holds because $\ljOperation(\ljVal_0, \ljVal_1)\joinType\ljType_{0}\joinType\ljType_{1} \equivTypeSub
		 \ljOperation(\tilde{\ljVal}_0, \tilde{\ljVal}_1)\joinType\tilde{\ljType}_{0}\joinType\tilde{\ljType}_{1}$.


	  \item[Case] \Rule{\LDJFunctionCreation}: $\ljExp\equiv\ljFunc \ljVar.\ljExp$ where location $\ljLocation$ with $\ljLocation\notin \dom(\ljHeap)$
		 and location $\tilde{\ljLocation}$ with $\tilde{\ljLocation}\notin \dom(\tilde{\ljHeap})$ extends the renaming $\bijection$ with
		 $\bijection'=\bijection[\ljLocation\mapsto\tilde{\ljLocation}]$. By definition \eqref{def:k-equivalence_value}
		 $\ljLocation:\ljType\equivTypeSubPrime\tilde{\ljLocation}:\ljType$ holds. Claim holds because:
		 \begin{align}
			\ljHeap[\ljLocation\mapsto\langle \ljEnv, \ljFunc \ljVar.\ljExp \rangle]
			\equivTypeSubPrime
			\tilde{\ljHeap}[\tilde{\ljLocation}\mapsto\langle \tilde{\ljEnv}, \ljFunc \ljVar.\tilde{\ljExp} \rangle]
		 \end{align}


	  \item[Case] \Rule{\LDJObjectCreation}: $\ljExp\equiv\ljNew~\ljExp_{p}$ where location $\ljLocation$ with $\ljLocation\notin \dom(\ljHeap)$
		 and location $\tilde{\ljLocation}$ with $\tilde{\ljLocation}\notin \dom(\tilde{\ljHeap})$ extends the renaming $\bijection$ with
		 $\bijection'=\bijection[\ljLocation\mapsto\tilde{\ljLocation}]$. By definition \eqref{def:k-equivalence_value}
		 $\ljLocation:\ljType\equivTypeSubPrime\tilde{\ljLocation}:\ljType$ holds.
		 By inversion 
		 \begin{align}
			\forall \ljSource \notin \bar{\ljType}:~ \tilde{\ljHeap},\tilde{\ljEnv},\ljType ~\entails~
			\ljExp_{p}[\ljSource\mapsto\tilde{\ljExp}] ~\eval~
			\tilde{\ljHeap}' ~|~ \tilde{\ljVal}:\tilde{\ljType}_{\ljVal}
		 \end{align}
		 holds. Claim holds because:
		 \begin{align}
			\ljHeap'[\ljLocation\mapsto\ljVal] \equivTypeSubPrime \tilde{\ljHeap}'[\tilde{\ljLocation}\mapsto\tilde{\ljVal}]
		 \end{align}


	  \item[Case] \Rule{\LDJFunctionApplication}: $\ljExp\equiv\ljExp_0(\ljExp_1)$.
		 By inversion 
		 \begin{align}
			&\forall \ljSource \notin \bar{\ljType}:~ \tilde{\ljHeap},\tilde{\ljEnv},\ljType ~\entails~
			\ljExp_{0}[\ljSource\mapsto\tilde{\ljExp}] ~\eval~
			\tilde{\ljHeap}' ~|~ \tilde{\ljLocation}:\tilde{\ljType}_{0}\\
			&\forall \ljSource \notin \bar{\ljType}:~ \tilde{\ljHeap}',\tilde{\ljEnv},\ljType ~\entails~
			\ljExp_{1}[\ljSource\mapsto\tilde{\ljExp}] ~\eval~ \tilde{\ljHeap}'' ~|~ \tilde{\ljVal}_{1}:\tilde{\ljType}_{1}
		 \end{align}
		 holds. $\forall\ljSource\notin\bar{\ljType}$ yields
		 \begin{description}
			\item[Subcase] $\ljSource \notin \ljType_{0} ~\wedge~
			   \ljSource \notin \tilde{\ljType}_{0}$:
			   By definition \eqref{def:k-equivalence_value} we obtain
			   \begin{align}
				  \ljSource \notin \ljType_{0} ~\wedge~ \ljSource \notin \tilde{\ljType}_{0} ~&\rightarrow~
				  \bijection(\ljLocation) = \tilde{\ljLocation}
			   \end{align}
			   and this leads by definition \eqref{def:k-equivalence_object}, \eqref{def:k-equivalence_heap} to:
			   \begin{align}
				  \begin{split}
					 &\langle \dot{\ljEnv}, \ljFunc \ljVar.\dot{\ljExp} \rangle=\ljHeap(\ljLocation) ~\wedge~
					 \langle \ddot{\ljEnv}, \ljFunc \ljVar.\ddot{\ljExp} \rangle=\tilde{\ljHeap}(\tilde{\ljLocation})\\
					 &\rightarrow~ \dot{\ljEnv} \equivTypeSub \ddot{\ljEnv} ~\wedge~
					 \ljFunc \ljVar.\dot{\ljExp} \equivTypeSub \ljFunc \ljVar.\ddot{\ljExp}
				  \end{split}
			   \end{align}
			   Hence, $\dot{\ljEnv}[\ljVar\mapsto\ljVal_{1}:\ljType_{1}] \equivTypeSub \ddot{\ljEnv}[\ljVar\mapsto\tilde{\ljVal}_{1}:\tilde{\ljType}_{1}]$
			   and $\dot{\ljExp} \equivTypeSub \ddot{\ljExp} ~\rightarrow~ \dot{\ljExp} \equivTypeSub\ddot{\ljExp}[\ljSource\mapsto\tilde{\ljExp}]$ complies
			   the postcondition to apply the induction where
			   \begin{align}
				  \begin{split}
					 &\forall \ljSource \notin \bar{\ljType}:~
					 \tilde{\ljHeap}'',\ddot{\ljEnv}[\ljVar\mapsto\tilde{\ljVal}_{1}:\tilde{\ljType}_{1}],\ljType ~\entails~\\
					 &\ddot{\ljExp} [\ljSource\mapsto\tilde{\ljExp}] ~\eval~ \tilde{\ljHeap}'' ~|~ \tilde{\ljVal}:\tilde{\ljType}_{\ljVal}
				  \end{split}
			   \end{align}

			\item[Subcase] $\ljSource \in \ljType_{0} ~\vee~
			   \ljSource \in \tilde{\ljType}_{0}$:
			   By lemma \eqref{thm:context_dependency} applied on $\ljType\joinType\ljType_{0} ~\subseteq~ \ljType_{\ljVal}$
			   and $\ljType\joinType\tilde{\ljType}_{0} ~\subseteq~ \tilde{\ljType}_{\ljVal}$ the claim holds because:
			   \begin{align}
				  &\ljVal:\ljType_{\ljVal} \equivTypeSub \tilde{\ljVal}:\tilde{\ljType}_{\ljVal}\\
				  &\ljHeap''' \equivTypeSub \tilde{\ljHeap}'''
			   \end{align}
		 \end{description}


	  \item[Case] \Rule{\LDJPropertyReference}: $\ljExp\equiv\ljExp_0[\ljExp_1]$.
		 By inversion 
		 \begin{align}
			&\forall \ljSource \notin \bar{\ljType}:~ \tilde{\ljHeap},\tilde{\ljEnv},\ljType ~\entails~ \ljExp_{0}[\ljSource\mapsto\tilde{\ljExp}] ~\eval~ \tilde{\ljHeap}' ~|~ \tilde{\ljLocation}:\tilde{\ljType}_{\ljLocation}\\
			&\forall \ljSource \notin \bar{\ljType}:~ \tilde{\ljHeap}',\tilde{\ljEnv},\ljType ~\entails~ \ljExp_{1}[\ljSource\mapsto\tilde{\ljExp}] ~\eval~ \tilde{\ljHeap}'' ~|~ \tilde{\ljStr}:\tilde{\ljType}_{\ljStr}
		 \end{align}
		 hold.

		 \begin{description}
			\item[Subcase] $\ljSource \notin \ljType_{\ljLocation} ~\wedge~ \ljSource \notin \tilde{\ljType}_{\ljLocation} ~\wedge~ \ljSource \notin \ljType_{\ljStr} ~\wedge~ \ljSource \notin \tilde{\ljType}_{\ljStr}$: By definition \eqref{def:k-equivalence_value} we obtain
			   \begin{align}
				  \ljSource \notin \ljType_{\ljLocation} ~\wedge~ \ljSource \notin \tilde{\ljType}_{\ljLocation} ~&\rightarrow~
				  \bijection(\ljLocation)=\tilde{\ljLocation}\\
				  \ljSource \notin \ljType_{\ljStr} ~\wedge~ \ljSource \notin \tilde{\ljType}_{\ljStr} ~&\rightarrow~
				  \ljStr=\tilde{\ljStr}
			   \end{align}
			   which follows that 
			   \begin{align}
				  \begin{split}
					 &\ljHeap''(\ljLocation)(\ljStr)\joinType\ljType_{\ljLocation}\joinType\ljType_{\ljStr} \equivTypeSub\\
					 &\tilde{\ljHeap}''(\tilde{\ljLocation})(\tilde{\ljStr})\joinType\tilde{\ljType}_{\ljLocation}\joinType\tilde{\ljType}_{\ljStr}
				  \end{split}
			   \end{align}

			\item[Subcase] $\ljSource \in \ljType_{\ljLocation} ~\vee~ \ljSource \in \tilde{\ljType}_{\ljLocation} ~\vee~ \ljSource \in \ljType_{\ljStr} ~\vee~ \ljSource
			   \in \tilde{\ljType}_{\ljStr}$:
			   By definition \eqref{def:k-equivalence_value} we obtain
			   \begin{align}
				  \begin{split}
					 &\ljHeap''(\ljLocation)(\ljStr)\joinType\ljType_{\ljLocation}\joinType\ljType_{\ljStr} \equivTypeSub\\
					 &\tilde{\ljHeap}''(\tilde{\ljLocation})(\tilde{\ljStr})\joinType\tilde{\ljType}_{\ljLocation}\joinType\tilde{\ljType}_{\ljStr}
				  \end{split}
			   \end{align}
		 \end{description}
		 Claim holds because $\ljHeap''(\ljLocation)(\ljStr)\joinType\ljType_{\ljLocation}\joinType\ljType_{\ljStr} \equivTypeSub
		 \tilde{\ljHeap}''(\tilde{\ljLocation})(\tilde{\ljStr})\joinType\tilde{\ljType}_{\ljLocation}\joinType\tilde{\ljType}_{\ljStr}$.


	  \item[Case] \Rule{\LDJPropertyAssignment}: $\ljExp\equiv\ljExp_0[\ljExp_1] = \ljExp_2$.
		 By inversion 
		 \begin{align}
			&\forall \ljSource \notin \bar{\ljType}:~ \tilde{\ljHeap},\tilde{\ljEnv},\ljType ~\entails~
			\ljExp_{0}[\ljSource\mapsto\tilde{\ljExp}] ~\eval~
			\tilde{\ljHeap}' ~|~ \tilde{\ljLocation}:\tilde{\ljType}_{\ljLocation}\\
			&\forall \ljSource \notin \bar{\ljType}:~ \tilde{\ljHeap}',\tilde{\ljEnv},\ljType ~\entails~
			\ljExp_{1}[\ljSource\mapsto\tilde{\ljExp}] ~\eval~
			\tilde{\ljHeap}'' ~|~ \tilde{\ljStr}:\tilde{\ljType}_{\ljStr}\\
			&\forall \ljSource \notin \bar{\ljType}:~ \tilde{\ljHeap}'',\tilde{\ljEnv},\ljType ~\entails~
			\ljExp_{2}[\ljSource\mapsto\tilde{\ljExp}] ~\eval~
			\tilde{\ljHeap}''' ~|~ \tilde{\ljVal}:\tilde{\ljType}_{\ljVal}
		 \end{align}
		 holds. $\forall\ljSource\notin\bar{\ljType}$ yields
		 \begin{description}
			\item[Subcase] $\ljSource \notin \ljType_{\ljLocation} ~\wedge~
			   \ljSource \notin \tilde{\ljType}_{\ljLocation} ~\wedge~
			   \ljSource \notin \ljType_{\ljStr} ~\wedge~
			   \ljSource \notin \tilde{\ljType}_{\ljStr} ~\wedge~
			   \ljSource \notin \ljType_{\ljVal} ~\wedge~
			   \ljSource \notin \tilde{\ljType}_{\ljVal}$.
			   By definition \eqref{def:k-equivalence_value}, \eqref{def:k-equivalence_heap} we obtain
			   \begin{align}
				  \ljSource \notin \ljType_{\ljLocation} ~\wedge~ \ljSource \notin \tilde{\ljType}_{\ljLocation} ~&\rightarrow~
				  \bijection(\ljLocation)=\tilde{\ljLocation}\\
				  \ljSource \notin \ljType_{\ljStr} ~\wedge~ \ljSource \notin \tilde{\ljType}_{\ljStr} ~&\rightarrow~
				  \ljStr=\tilde{\ljStr}\\
				  \ljSource \notin \ljType_{\ljVal} ~\wedge~ \ljSource \notin \tilde{\ljType}_{\ljVal} ~&\rightarrow~
				  \ljVal=\tilde{\ljVal}
			   \end{align}
			   and this leads to:
			   \begin{align}
				  \begin{split}
					 &\ljHeap'''[\ljLocation,\ljStr\mapsto \ljVal:\ljType_{\ljVal} \joinType \ljType_{\ljLocation} \joinType \ljType_{\ljStr}] \equivTypeSub\\
					 &\tilde{\ljHeap}'''[\tilde{\ljLocation},\tilde{\ljStr}\mapsto \tilde{\ljVal}:\tilde{\ljType}_{\ljVal} \joinType \tilde{\ljType}_{\ljLocation} \joinType \tilde{\ljType}_{\ljStr}]
				  \end{split}
			   \end{align}

			\item[Subcase] $\ljSource \in \ljType_{\ljLocation} ~\vee~
			   \ljSource \in \tilde{\ljType}_{\ljLocation} ~\vee~
			   \ljSource \in \ljType_{\ljStr} ~\vee~
			   \ljSource \in \tilde{\ljType}_{\ljStr} ~\vee~
			   \ljSource \in \ljType_{\ljVal} ~\vee~
			   \ljSource \in \tilde{\ljType}_{\ljVal}$.
			   By definition \eqref{def:k-equivalence_value}, \eqref{def:k-equivalence_heap}
			   \begin{align}
				  \begin{split}
					 &\ljHeap'''[\ljLocation,\ljStr\mapsto \ljVal:\ljType_{\ljVal} \joinType \ljType_{\ljLocation} \joinType \ljType_{\ljStr}] \equivTypeSub\\
					 &\tilde{\ljHeap}'''[\tilde{\ljLocation},\tilde{\ljStr}\mapsto \tilde{\ljVal}:\tilde{\ljType}_{\ljVal} \joinType \tilde{\ljType}_{\ljLocation} \joinType \tilde{\ljType}_{\ljStr}]
				  \end{split}
			   \end{align}
			   holds.
		 \end{description}
		 Claim holds because $\ljHeap'''[\ljLocation,\ljStr\mapsto \ljVal:\ljType_{\ljVal} \joinType \ljType_{\ljLocation} \joinType \ljType_{\ljStr}] \equivTypeSub
		 \tilde{\ljHeap}'''[\tilde{\ljLocation},\tilde{\ljStr}\mapsto \tilde{\ljVal}:\tilde{\ljType}_{\ljVal} \joinType \tilde{\ljType}_{\ljLocation} \joinType \tilde{\ljType}_{\ljStr}]$.


	  \item[Case] \Rule{\LDJConditionTrue}: $\ljExp\equiv\ljIf~ (\ljExp_0)~ \ljExp_1,~ \ljExp_2$.
		 By inversion 
		 \begin{align}
			&\forall \ljSource \notin \bar{\ljType}:~ \tilde{\ljHeap},\tilde{\ljEnv},\ljType ~\entails~
			\ljExp_{0}[\ljSource\mapsto\tilde{\ljExp}] ~\eval~
			\tilde{\ljHeap}' ~|~ \tilde{\ljVal}_{0}:\tilde{\ljType}_{0}\\
			\begin{split}
			&\forall \ljSource \notin \bar{\ljType}:~ \tilde{\ljHeap}',\tilde{\ljEnv},\ljType\joinType\tilde{\ljType}_{0} ~\entails~\\
			&\ljExp_{1}[\ljSource\mapsto\tilde{\ljExp}] ~\eval~
			\tilde{\ljHeap}_{1}'' ~|~ \tilde{\ljVal}_{1}:\tilde{\ljType}_{1}
			\end{split}
		 \end{align}
		 holds.  $\forall\ljSource\notin\bar{\ljType}$ yields
		 \begin{description}
			\item[Subcase] $\ljSource \notin \ljType_{0} ~\wedge~
			   \ljSource \notin \tilde{\ljType}_{0}$: 
			   By definition \eqref{def:k-equivalence_value} we obtain
			   \begin{align}
				  \ljSource \notin \ljType_{0} ~\wedge~ \ljSource \notin \tilde{\ljType}_{0} ~&\rightarrow~ \ljVal_{0}=\tilde{\ljVal}_{0}
			   \end{align}
			   Claim holds by inversion of \Rule{\LDJConditionTrue}.

			\item[Subcase] $\ljSource \in \ljType_{\ljVal} ~\vee~
			   \ljSource \in \tilde{\ljVal}_{0}$:
			   By lemma \eqref{thm:context_dependency} applied on $\ljType\joinType\ljType_{0} ~\subseteq~ \ljType_{1}$ and 
			   $\ljType\joinType\tilde{\ljType}_{0} ~\subseteq~ \tilde{\ljType}_{1}$ the claim holds because:
			   \begin{align}
				  \ljVal_{1}:\ljType_{1} \equivTypeSub \tilde{\ljVal}_{1}:\tilde{\ljType}_{1}\\
				  \ljHeap''_{1} \equivTypeSub \tilde{\ljHeap}''_{1}
			   \end{align}

		 \end{description}


	  \item[Case] \Rule{\LDJConditionFalse}: Analogous to case \\ \Rule{\LDJConditionTrue}.

	  \item[Case] \Rule{\LDJTrace}: $\ljExp\equiv\ljTrace~ (\ljExp)$.
		 By inversion 
		 \begin{align}
			\forall \ljSource' \notin \bar{\ljType}:~ \tilde{\ljHeap},\tilde{\ljEnv},\ljType\joinType\ljSource ~\entails~
			\ljExp_{\ljSource}[\ljSource\mapsto\tilde{\ljExp}] ~\eval~
			\tilde{\ljHeap}' ~|~ \tilde{\ljVal}:\tilde{\ljType}_{\ljVal}
		 \end{align}
		 holds.  $\forall\ljSource\notin\bar{\ljType}$ yields
		 \begin{description}
			\item[Subcase] $\ljSource' \neq \ljSource$:
			   Claim holds by inversion \Rule{\LDJTrace}.

			\item[Subcase] $\ljSource' = \ljSource$:
			   By definition \eqref{def:substitution-i} we obtain:
			   \begin{align}
				  \forall \ljSource' \notin \bar{\ljType}:~ \tilde{\ljHeap},\tilde{\ljEnv},\ljType\joinType\ljSource ~\entails~
				  \tilde{\ljExp} ~\eval~
				  \tilde{\ljHeap}'' ~|~ \tilde{\ljVal}':\tilde{\ljType}_{\ljVal}'
			   \end{align}
			   By lemma \eqref{thm:context_dependency} applied on $\ljType\joinType\ljSource ~\subseteq~ \tilde{\ljType}_{\ljVal}'$ the claim holds because:
			   \begin{align}
				  \ljSource\in\tilde{\ljType}_{\ljVal}' ~&\rightarrow~ \ljVal:\ljType_{\ljVal} \equivTypeSub \tilde{\ljVal}':\tilde{\ljType}_{\ljVal}'\\
				  \ljSource\in\tilde{\ljType}_{\ljVal}' ~&\rightarrow~ \ljHeap' \equivTypeSub \tilde{\ljHeap}''
			   \end{align}
		 \end{description}

   \end{description}
\end{proof}


\section{Correctness}
\label{sec:proof_correctness}

First, we state some auxiliary lemmas.

\begin{lemma}[Subset consistency on states]\label{thm:subset_consistency_state}
   $\ljHeap \equivLattice \aState_{0}$ implies $\ljHeap \equivLattice \aState_{0} \join \aState_{1}$ with $\aState_{0} \sqsubseteq \aState_{0} \join \aState_{1}$.
\end{lemma}

\begin{proof}[Proof of Lemma \ref{thm:subset_consistency_state}]
   By definition \eqref{def:consistency_relation}.
\end{proof}

\begin{lemma}[Subset consistency on values]\label{thm:subset_consistency_value}
   $\ljTypedVal \equivLattice \aVal_{0}$ implies $\ljTypedVal \equivLattice \aVal_{0} \join \aVal_{1}$ with $\aVal_{0} \sqsubseteq \aVal_{0} \join \aVal_{1}$.
\end{lemma}

\begin{proof}[Proof of Lemma \ref{thm:subset_consistency_value}]
   By definition \eqref{def:consistency_relation}. 
\end{proof}

\begin{lemma}[Consistency on dependencies]\label{thm:dependencies_consistency}
   $\ljVal:\ljType_{\ljVal} \equivLattice \aVal ~\wedge~ \ljType \equivLattice \D$ implies $\ljVal:\ljType_{\ljVal}\joinType\ljType \equivLattice
   \langle \aLattice_{\aVal}, \aObjlabel_{\aVal}, \D_{\aVal}\join\D \rangle$
\end{lemma}

\begin{proof}[Proof of Lemma \ref{thm:dependencies_consistency}]
   By definition \eqref{def:consistency_relation}.
\end{proof}

\begin{lemma}[Property Update]\label{thm:property_update_consistency}
   $\forall\ljStorable,\aStorable,\aLattice,\aVal ~|~ \ljStorable\equivLattice\aStorable$ implies $\ljStorable\equivLattice\aStorable[\aLattice\mapsto\aStorable(\aLattice)\join\aVal]$
\end{lemma}

\begin{proof}[Proof of Lemma \ref{thm:property_update_consistency}]
   By definition \eqref{def:consistency_relation}.
\end{proof}

\begin{proof}[Proof of Theorem \ref{thm:correctness}]

   Suppose that $\ljHeap,\ljEnv,\ljType ~\entails~ \ljExp ~\eval~ \ljHeap' ~|~ \ljVal$ then
   $\forall \aState,\aScope$ with $\ljHeap\equivLattice\aState$, $\ljEnv\equivLattice\aScope$ and
   $\ljType\equivLattice\D_{\aState}$: $\aState,\aScope ~\entails~ \ljExp ~\eval~ \aState' ~|~ \aVal$ with
   $\ljHeap'\equivLattice\aState'$ and $\ljVal\equivLattice\aVal$. Proof by induction on the derivation of
   $\ljHeap,\ljEnv,\ljType ~\entails~ \ljExp ~\eval~ \ljHeap' ~|~ \ljTypedVal$ and 
   $\aState,\aScope ~\entails~ \ljExp ~\eval~ \aState' ~|~ \aVal$.

   \begin{description}


	  \item[Case] \Rule{\LDJConstant}: $\ljExp\equiv\ljConst$.
		 Claim holds because $\ljConst:\ljType \equivLattice \langle \ljConst, \emptyset, \D_{\aState} \rangle$.


	  \item[Case] \Rule{\LDJVariable}: $\ljExp\equiv\ljVar$.
		 Claim holds because $\ljEnv \equivLattice \aScope$ implies $\ljEnv(\ljVar)\joinType\ljType \equivLattice \aScope(\ljVar)\join \D_{\aState}$.


	  \item[Case] \Rule{\LDJOperation}: $\ljExp\equiv\ljOp(\ljExp_{0}, \ljExp_{1})$
		 where by inversion
		 \begin{align}
			\aState,\aScope ~&\entails~ \ljExp_0 ~\eval~ \aState' ~|~ \aVal_{0}\\
			\aState',\aScope ~&\entails~ \ljExp_1 ~\eval~ \aState'' ~|~ \aVal_{1}
		 \end{align}
		 holds. Claim holds by definition \eqref{def:abstract_operation}.


	  \item[Case] \Rule{\LDJFunctionCreation}: $\ljExp\equiv\ljFunc \ljVar.\ljExp$.
		 Claim holds by definition \eqref{def:consistency_relation}.

	  \item[Case] \Rule{\LDJObjectCreation}: $\ljExp\equiv\ljNew~\ljExp_{p}$
		 where by inversion
		 \begin{align}
			\aState,\aScope ~&\entails~ \ljExp_{p} ~\eval~ \aState' ~|~ \aVal
		 \end{align}
		 holds. Claim holds by definition \eqref{def:consistency_relation}.


	  \item[Case] \Rule{\LDJFunctionApplication}: $\ljExp\equiv\ljExp_0(\ljExp_1)$
		 where by inversion 
		 \begin{align}
			\aState,\aScope ~&\entails~ \ljExp_0 ~\eval~ \aState' ~|~ \langle \aLattice_{0}, \aObjlabel_{0}, \D_{0} \rangle\\
			\aState', \aScope ~&\entails~ \ljExp_1 ~\eval~ \aState'' ~|~ \aVal_{1}
		 \end{align}
		 holds. By definition \eqref{def:consistency_relation} we obtain $\exists\ljSLocation\in\aObjlabel_{0} ~|~ \ljLocation\equivLattice\ljSLocation$.
		 \begin{description}
			\item[Subcase] $\forall\ljSLocation\in\aObjlabel_{0}: \ljLocation\equivLattice\ljSLocation$
			   where by inversion
			   \begin{align}
				  \begin{split}
					 &\aState''[\D\mapsto\D_{\aState''}\join\D_{0}],\dot{\aScope}[\ljVar\mapsto\aVal_{1}] ~\entails~ \\
					 &\ljExp ~\eval~ \aState''' ~|~ \aVal
				  \end{split}
			   \end{align}
			   holds.
			   \begin{description}
				  \item[Subsubcase] $\langle\aState''[\D\mapsto\D_{\aState''}\join\D_{0}],\aVal_{1}\rangle \sqsubseteq \aStore(\ljSLocation)_{\aFuncInput}$.
					 Claim holds by lemma \eqref{thm:subset_consistency_state} and \eqref{thm:function_store}. 

				  \item[Subsubcase] $\langle\aState''[\D\mapsto\D_{\aState''}\join\D_{0}],\aVal_{1}\rangle \not\sqsubseteq \aStore(\ljSLocation)_{\aFuncInput}$.
					 Claim holds by lemma \eqref{thm:subset_consistency_state} and \eqref{thm:subset_consistency_value}.

			   \end{description}

			\item[Subcase] $\forall\ljSLocation\in\aObjlabel_{0}:\ljLocation\not\equivLattice\ljSLocation$
			   Case holds by $\exists\ljSLocation\in\aObjlabel_{0} ~|~ \ljLocation\equivLattice\ljSLocation$ and
			   lemma \eqref{thm:subset_consistency_state} and \eqref{thm:subset_consistency_value}.

		 \end{description}


	  \item[Case] \Rule{\LDJPropertyReference}: $\ljExp\equiv\ljExp_0[\ljExp_1]$
		 where by inversion 
		 \begin{align}
			\aState,\aScope ~&\entails~ \ljExp_0 ~\eval~ \aState' ~|~ \langle \aLattice_{0}, \aObjlabel_{0}, \D_{0} \rangle\\
			\aState',\aScope ~&\entails~ \ljExp_1 ~\eval~ \aState'' ~|~ \langle \lvString, \aObjlabel_{1}, \D_{1} \rangle 
		 \end{align}
		 holds. Claim holds because lemma \eqref{thm:property_reference} and definition \eqref{def:consistency_relation} implies
		 $\ljHeap''(\ljLocation)(\ljStr) ~\equivLattice~ \aVal$.


	  \item[Case] \Rule{\LDJPropertyAssignment}: $\ljExp\equiv\ljExp_0[\ljExp_1] = \ljExp_2$
		 where by inversion 
		 \begin{align}
			\aState,\aScope ~&\entails~ \ljExp_0 ~\eval~ \aState' ~|~ \langle \aLattice_{0}, \aObjlabel_{0}, \D_{0} \rangle \\
			\aState',\aScope ~&\entails~ \ljExp_1 ~\eval~ \aState'' ~|~ \langle \lvString, \aObjlabel_{1}, \D_{1} \rangle \\
			\aState'',\aScope ~&\entails~ \ljExp_2 ~\eval~ \aState''' ~|~ \aVal
		 \end{align}
		 holds. Claim holds because lemma \eqref{thm:property_assignment} and definition \eqref{def:consistency_relation} implies
		 $\ljHeap'''(\ljLocation)[\ljStr\mapsto\ljVal:\ljType_{\ljVal}] ~\equivLattice~ \aState''''$.


	  \item[Case] \Rule{\LDJConditionTrue}: $\ljExp\equiv\ljIf~ (\ljExp_0)~ \ljExp_1,~ \ljExp_2$
		 where by inversion 
		 \begin{align}
			&\aState,\aScope ~\entails~ \ljExp_0 ~\eval~ \aState' ~|~ \langle \aLattice_{0}, \aObjlabel_{0}, \D_{0} \rangle \\
			&\aState'[\D \mapsto \D_{\aState} \join \D_{0}],\aScope ~\entails~ \ljExp_1 ~\eval~ \aState_{1}'' ~|~ \aVal_1
		 \end{align}

		 holds. By definition \eqref{def:consistency_relation} we obtain $\ljTrue ~\equivLattice~ \aLattice_{0}$.
		 \begin{description}
			\item[Subcase]$\aLattice_{0}= \lvTrue$. Claim holds by inversion.
			\item[Subcase]$\aLattice_{0}= \lvBool$. Claim holds by
			   lemma \eqref{thm:subset_consistency_state} and \eqref{thm:subset_consistency_value}.

		 \end{description}


	  \item[Case] \Rule{\LDJConditionFalse}: Analogous to case \\ \Rule{\LDJConditionTrue}.

	  \item[Case] \Rule{\LDJTrace}: $\ljExp\equiv\ljTrace~ (\ljExp)$
		 where by inversion
		 \begin{align}
			\aState[\D\mapsto\D_{\aState}\join\ljSource], \aScope ~\entails~ \ljExp ~\eval~ \aState' ~|~ \aVal
		 \end{align}
		 holds. Claim holds by definition \eqref{def:consistency_relation}.

   \end{description}
\end{proof}


\section{Termination}
\label{sec:proof_termination}


\begin{lemma}[Monotonicity]\label{thm:monotony}
   The inference system in Fig.~\ref{fig:inference-rules_abstract} is monotone.
\end{lemma}

\begin{proof}[Proof of Lemma \ref{thm:monotony}]
   We have to show monotonicity, that is,\\
   $\forall \aState, \tilde{\aState}, \aScope, \tilde{\aScope}, \aVal, \tilde{\aVal}, \ljExp$ with $\aState\sqsubseteq\tilde{\aState}$,
   $\aScope\sqsubseteq\tilde{\aScope}$:    
   \begin{align}
	  \begin{split}
		 &\aState,\aScope ~\entails~ \ljExp ~\eval~ \aState' ~|~ \aVal ~\wedge~
		 \tilde{\aState}, \tilde{\aScope} ~\entails~ \ljExp ~\eval~ \tilde{\aState}' ~|~ \tilde{\aVal}\\
		 &~\rightarrow~ \aState'\sqsubseteq\tilde{\aState}' ~\wedge~ \aVal\sqsubseteq\tilde{\aVal}
	  \end{split}
   \end{align}
   The proof is by induction on $\aState,\aScope
   ~\entails~ \ljExp ~\eval~ \aState' ~|~ \aVal$. 
   It is easy to see that each rule preserves
   monotonicity if the recursive uses of the judgment do
   so, too.
\end{proof}


\begin{lemma}[Ascending Chain Condition]
   \label{thm:ascending-chain-condition}
   For each subject program, the analysis reaches a finite subset of
   the partially ordered set $(\AAnalysisLattice,\sqsubseteq)$. 
\end{lemma}

\begin{proof}[Proof sketch of Lemma \ref{thm:ascending-chain-condition}]
   Inspection of Fig.~\ref{fig:abstract-semantic-domains} shows that
   the finiteness of many components of the abstract semantic domain depend on the finiteness of the
   number of labels $\ell$ in a program. The sole exception is the
   abstract object, which is a mapping from lattice values to abstract
   values. Because the $\LvString$ sublattice of
   $\ALatticeValue$ has infinitely many values, the domain of abstract
   objects is infinite and has infinite height.

   We argue that each abstract object only uses a finite number of its
   fields by looking at the possible arguments to get and put
   properties. Non-$\top$ values of $\LvString$ are generated by
   constants in the program and by primitive operations applied to
   constants. Disregarding recursion, there are only finitely many such
   constants and operations, so that only a finite sublattice of
   $\AObject$ is exercised.

   Adding recursion does not change this picture because the
   function store subsumes function calls. If a call to a
   function has a fixed string argument that is different to the
   argument of previous calls, then the analysis subsumes these
   arguments to $\top$ and analyzes the function body with the subsumed
   argument.
\end{proof}


\begin{proof}[Proof of Theorem \ref{thm:termination}]
   For each $\aState$, $\aScope$, and $\ljExp$, there exist $\aState'$ and $\aVal$
   such that $\aState,\aScope ~\entails~ \ljExp ~\eval~ \aState' ~|~ \aVal$.
   The computation terminates with a fixpoint because all inference rules are monotone
   (Lemma \ref{thm:monotony}) and the analysis lattice has finite height (Lemma \ref{thm:ascending-chain-condition}).
\end{proof}


\end{document}